\documentclass{article}

\usepackage[utf8]{inputenc}

\usepackage{graphicx}
\usepackage{bm,amsmath,amssymb}
\usepackage{algorithm}
\usepackage{algpseudocode}
\usepackage[noabbrev]{cleveref}
\usepackage{array}
\usepackage{float}
\usepackage{url}
\usepackage{subcaption}
\newtheorem{proposition}{Proposition}
\newcounter{noqed}
\newcommand{\qed}{ \ifmmode\text{ }\fi\rule[-.05em]{.3em}{.7em}\setcounter{noqed}{0}}
\newenvironment{proof}{\par\smallskip\noindent{\bf Proof.}\setcounter{noqed}{1}}{\ifnum\value{noqed}=1\qed\fi\par\medskip}

\advance\textheight by 2cm
\advance\topmargin -1cm

\advance\textwidth by 3cm
\advance\oddsidemargin by -1.5cm
\advance\evensidemargin by -1.5cm

\newcommand{\?}{\mskip1.5mu}
\newcommand{\rank}{\operatorname{\textsf{rank}}}
\newcommand{\select}{\operatorname{\textsf{select}}}

\newcommand{\add}{\operatorname{\textsf{add}}}
\newcommand{\push}{\operatorname{\textsf{push}}}
\newcommand{\pop}{\operatorname{\textsf{pop}}}
\newcommand{\prefix}{\operatorname{\textsf{prefix}}}
\newcommand{\find}{\operatorname{\textsf{find}}}
\newcommand{\findc}{\overline{\operatorname{\textsf{find}}}}

\newcommand{\N}{\mathbf N}

\def\..{\, \mathpunct{\ldotp\ldotp}} 

\newcolumntype{x}[1]{>{\centering\arraybackslash\hspace{0pt}}p{#1}}

\begin{document}
\bibliographystyle{plain}

\title{Compact Fenwick trees for dynamic ranking and selection}

\author{Stefano Marchini\and Sebastiano Vigna\\Dipartimento di Informatica, Universit\`a degli
Studi di Milano, Italy}

\maketitle

\begin{abstract}
The Fenwick tree~\cite{FenNDSCFT} is a classical implicit
data structure that stores an array in such a way that modifying an element,
accessing an element, computing a prefix sum and performing a predecessor search
on prefix sums all take logarithmic time. We introduce a number of variants
which improve the classical implementation of the tree: in particular, we can
reduce its size when an upper bound on the array element is known, and we can
perform much faster predecessor searches. Our aim is to use our variants to
implement an efficient \emph{dynamic bit vector}:
our structure is able to perform updates, ranking and selection in logarithmic
time, with a space overhead in the order of a few percents, outperforming
existing data structures with the same purpose. Along the way, we highlight the
pernicious interplay between the arithmetic behind the Fenwick tree and the
structure of current CPU caches, suggesting simple solutions that improve
performance significantly.
\end{abstract}

\section{Introduction}

The problem of building static data structures which perform
rank and select operations on vectors of $n$ bits in constant time using additional
$o(n)$ bits has received a great deal of attention in the last two
decades starting form Jacobson's seminal work on \emph{succinct
data structures}.~\cite{JacSSTG} The \emph{rank} operator takes a position in
the bit vector and returns the number of preceding ones. The \emph{select} operation returns the position of the $k$-th one in the vector, given $k$. These two operations are at
the core of most existing succinct static data structure, as those representing
sets, trees, and so on. Another line of research studies analogous
\emph{compact} structures that are more practical, but they use $cn$
additional bits, for a usually small constant $c$.\cite{VigBIRSQ}

A much less studied problem is that of implementing \emph{dynamic} rank and
select operators. In this case we have again an underlying bit vector, which
however is dynamic, and it is possible to change its size and its content. Known
lower bounds~\cite{BeFOBPPRP} tell us that dynamic ranking and selection cannot
be performed faster than $\Omega(\lg n/\lg w)$, where $w$ is the word size. Some
theoretical data structures match this bound, but they are too complex to be implemented in practice
(e.g., they have unpredictable tests or memory-access patterns which are not cache-friendly).

It is a trivial observation that by keeping track in counters of the number of
ones in blocks of $q$ words rank and selection can be performed first on the
counters and then locally in each block. In particular, any dynamic data
structure that can keep track of such blocks and quickly provides prefix sums and
predecessor search in the prefix sums (i.e., find the last block whose prefix
sum is less than a given bound) can be used to implement dynamic ranking and
selection.

In this paper we consider the Fenwick tree,~\cite{FenNDSCFT} an implicit data
structure that was devised to store a sequence of integers, providing
logarithmic-time updates, computation of prefix sums, and
predecessor searches (into the prefix sums). Its original goal was efficient
dynamic maintenance of the numerosity of the symbols seen in a stream for
the purpose of performing arithmetic compression.

Our aim is to improve the Fenwick tree in general, keeping in mind the
idea of using it to implement a dynamic bit vector. To this purpose, first we
will show how to \emph{compress} the tree, using additional knowledge we might have on the
values stored in the tree, and introduce a \emph{complemented search} operation
that is necessary to implement selection on zeros; then, we propose a different,
level-order layout for the tree. The layout is very efficient
(cachewise) for predecessor search (and thus for selection), whereas
the classical Fenwick layout is more efficient
for prefix sums (and thus for ranking).

We then discuss the best way to use a Fenwick tree to support a dynamic bit
vector, and argue that due to the current CPU structure the tree should
keep track approximately of the number of ones in one or two cache lines, as ranking
and selection in a cache line can be performed at a very high speed using
specialized CPU instructions. Finally, we perform a wide range of experiments
showing the effectiveness of our approach, and compare it with previous
implementations.

All the code used in this paper is available from the authors under the
GNU Lesser General Public License, version 3 or later, as part of the
Sux project (\url{http://sux.di.unimi.it/}).

\section{Notation}

We use $w$ to denote the machine word size, $\lg$ for the binary logarithm,
$\&$, $|$ and $\oplus$ to denote bitwise and, or and xor on integers,
$\gg$ and $\ll$ for right and left shifting, and an overline for bitwise
negation (as in $\bar x$).
Following Knuth,\cite[7.1.3]{KnuACPIV} we use $\rho x$ to denote the \emph{ruler
function} of $x$, that is, the index (starting from zero) of the lowest bit set ($\rho0=\infty$), $\lambda x$ for
the index of the highest bit set (i.e., $\lambda x = \lfloor \lg x \rfloor$; $\lambda x$ is undefined when $x\leq 0$),
and $\nu x$ for the \emph{sideways sum} of $x\geq 0$, that is, the number of
bits set to one in the binary representation of $x$. Note the easy mnemonics:
$\rho x$ is the position of the \emph{rightmost} one, $\lambda x$ is the position of the \emph{leftmost} one, and $\nu x$ is the
\emph{number} of ones in the binary representation of $x$. Alternative common
names for these functions are LSB, MSB and \emph{population count}, respectively.

Let $\bm v=\bigl\langle v_1,v_n,\ldots,v_n\bigr\rangle$ be a vector of $n$
natural numbers indexed from one.
We define the operations
\begin{align}
\prefix_{\bm{v}}(p) &= \sum_{i=1}^{p} v_i & 0 \leq p \le n \label{get}\\
\find_{\bm{v}}(x) &= \max \bigl\{\? p \mid \prefix_{\bm v}(p) \le x\?\bigr\}
&x\in\N\label{find}
\end{align}
that is, the sum of the prefix of length $p$ of $\bm v$ and
the length of the longest prefix with sum less than or equal to $x$.

Finally, let $\bm b =\bigl\langle b_0,b_1,\ldots,b_{n-1}\bigr\rangle$ be an
array of $n$ bits indexed from zero.
We define
\begin{align*}
\rank_{\bm{b}}(p) &= \bigl| \bigl\{ i \in[0\..p)\mid  b_i = 1
\bigr\} \bigr| & 0 \leq p \le n\\
\select_{\bm{b}}(k) &= \max\bigl\{\?p \in[0\..n)\mid
\text{rank}_{\bm{b}}(p) \le k \?\bigr\} & 0 \le k < \text{rank}_{\bm{b}}(n)
\end{align*}
Thus, $\rank_{\bm{b}}(p)$ counts the number of ones up to position $p$
(excluded), while $\select_{\bm{b}}(k)$ returns the position of the $k$-th one,
indexed starting from zero.

We remark that our choice of indexing is driven by the data structures we
will describe: the Fenwick tree is easiest to describe using vectors indexed
from one, whereas ranking and
selection are much simpler to work with when the underlying bit vector is
indexed from zero.

\section{Related Work}

Several papers in the area of succinct data structures discuss the
\emph{Searchable Prefix Sum} problem, which is the same problem solved by the
Fenwick tree.\cite{RRRSDDS,MaNDECSFTI,HSSSDSSPS} However, as we discussed in the
introduction, these solutions, while providing strong theoretical guarantees,
do not yield practical improvements. Bille~\textit{et al.}~\cite{BCPSPSFT}
is the work in spirit most similar to ours, in that they study succinct representations
of Fenwick trees, extending moreover the construction beyond the binary case:
in particular, they study, as we do, level-order layouts. However, also in this
case the authors aim at asymptotic bounds, rather than dealing with the
practicalities of a tuned implementation.

\section{The Fenwick tree}
\label{sec:tree}

The Fenwick tree~\cite{FenNDSCFT} is an implicit dynamic data structure that
represents a list of $n$ natural numbers and provides logarithmic-time
\textsf{prefix} and \textsf{find} operations; also updating an element of the
list by adding a constant takes logarithmic time. There is an implicit
operation that retrieves the $k$-th value by computing $\prefix(k)-\prefix(k-1)$,
but Fenwick notes that this operation can
be implemented more efficiently (on average).

Strictly speaking, the Fenwick tree is not a tree; that is, it cannot be
univocally described as a set of nodes tied up by a single parent relationship.
There are three different
parent-child relationships (on the same set of nodes) that are useful for
different primitives (see Figure~\ref{fig:tree}).

More precisely, for a vector $\bm v$ of $n$ values
the content of the nodes of a Fenwick tree are gathered in a vector
of \emph{partial sums}~$\bm f$, again of size $n$, storing in position $j\in[1\..n]$ the sum of
values in $\bm v$ with index in $\bigl(j - 2^{\rho j}\..j\bigr]$. In particular, for odd
$j$ we have $f_j=v_j$.

We are going now to describe the three different
parent-child relationships. In doing so, we show an interesting duality law and
introduce new rules to perform prefix queries and updates in a descending
fashion, as opposed to the original ascending technique described by Fenwick.

\begin{figure}
\centering
\includegraphics{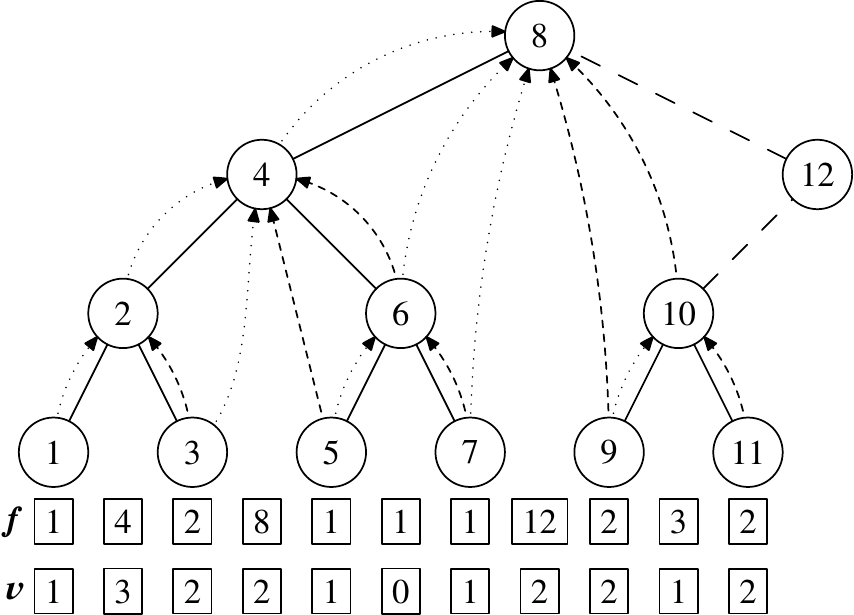}
\caption{\label{fig:tree}A Fenwick tree represented as a sideways
heap with $11$ nodes: note that we need to add an implicit node $12$ to provide
a parent for node $10$. The dashed arrows show the parent relationship of the
interrogation tree (omitting the root $0$), while the dotted arrows show the
parent relationship of the update tree (omitting the parent arrows with target
larger than $11$). The vector $\bm v$ contains the original values,
whereas the vector $\bm f$ represents implicitly the Fenwick tree: every value is
associated with the node above it.}
\end{figure}

\subsection{The search tree}

The \emph{search tree} is actually a
\emph{sideways heap}, a data structure which had been introduced previously by
Dov Harel,~\cite{HarEATBT} unknown to Fenwick. A detailed description is given
by Knuth,\cite[7.1.3]{KnuACPIV} who describes sideways heaps using an infinite set of
nodes: here we extend this description to the Fenwick tree, as it highlights a
number of beautiful symmetries and dualities. The finite structure of a Fenwick
tree of $n$ nodes is then obtained essentially by considering only the
parent-child relationships between the nodes $[1\..n]$.

The infinite sideways heap has an infinite number
of leaves given by the odd numbers. The parents of the leaves are the
odd multiples of $2$ (i.e., even numbers not divisible by $4$); their
grandparents are the odd multiples of $4$, and so on. Node $j$ is exactly $\rho
j$ levels above the leaf level.

More precisely, the parent of node $j$ is given by
\[\bigl(j - (1 \mathbin\ll \rho j)\bigr) \mathbin| \bigl(1 \mathbin\ll (\rho j +
1)\bigr),\] which in two's complement notation can be computed by
\[\bigl(j - (j \mathbin\& -j)\bigr) \mathbin| \bigl((j \mathbin\& -j)
\mathbin\ll 1\bigr)=\bigl(j \mathbin\& (j-1)\bigr) \mathbin| \bigl((j \mathbin\&
-j) \mathbin\ll 1\bigr).\]
The children of node $j$ are $j \pm\bigl(1\mathbin\ll (\rho
j - 1)\bigr)$ , which in two's complement notation is $j \pm \bigl((j \mathbin\&
-j) \mathbin\gg 1\bigr)$.

The infinite sideways heap has no root, so it is not technically a binary tree,
but if we restrict ourselves to a finite number $n$ of the form $2^k-1$
we obtain a perfect binary tree with root
$1\mathbin\ll\lambda n = 2^{\lambda n}$.
Otherwise, we have to add (implicitly) some nodes
beyond $n$ to connect the nodes $\bigl[2^{\lambda n} + 1\.. n\bigr]$ to the root: for example, in
Figure~\ref{fig:tree} we show (continuous lines) a sideways heap with $n=11$, but we have to add
node $12$ to connect the three-node right subtree to the root. When discussing
the height or depth of a node of a Fenwick tree we will always refer to the
associated sideways heap. Note that the heap has height $\lambda n$.


In terms of partial sums, a node $j$ at
height $h$ stores the sum of a subsequence of $2^h$ values of $\bm v$ that
correspond exactly to node $j$ plus its left subtree.
Nodes with index larger than $n$ \emph{do not store any partial sum}, and
indeed they are represented only implicitly.

\begin{figure}
\centering
\begin{minipage}[b]{7	cm}
  \begin{algorithmic}[1]
    \Function{$\mathrm{find}$}{$k$}
      \State $p \gets 0$
      \State $q \gets 1 \mathbin\ll \lambda n$
      \While{$q \neq 0$}
      \label{line:ignore}
        \If{$p + q \leq n$}
          \State $m \gets f[p+q]$
          \label{comp_val}
          \If{$k \ge m$}
            \State $p \gets p + q$
            \State $k \gets k - m$
          \EndIf
        \EndIf
        \State $q \gets q \mathbin\gg 1$
      \EndWhile
      \State \Return{$\langle p, k\rangle$}
    \EndFunction
  \end{algorithmic}
\end{minipage}\includegraphics{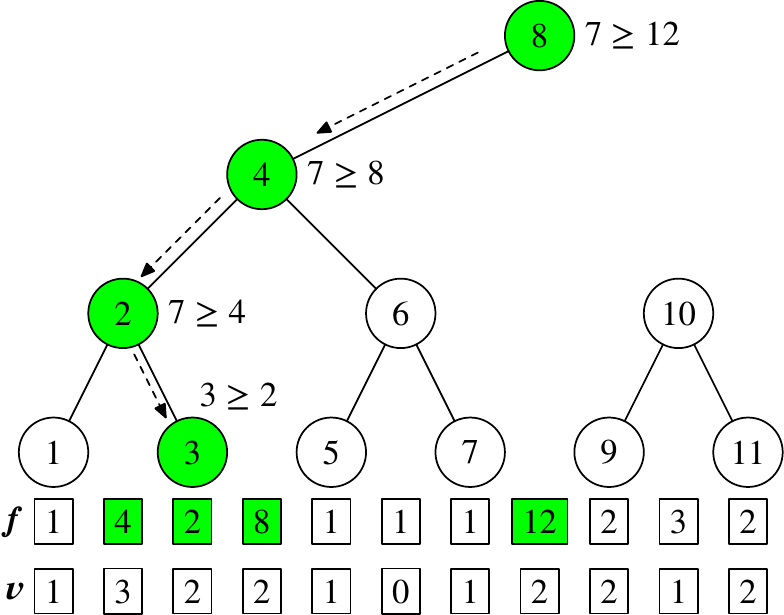}
\caption{\label{fig:find}Pseudocode for \textsf{find}, and a sample execution on the tree
of Figure~\ref{fig:tree}: we are looking for the rightmost position with prefix
sum at most $7$. We highlight the visited nodes and the associated elements in the
vector $\bm f$ representing the tree. The test is against the element of $\bm f$ associated with the node. The returned
pair is $\langle 3,1\rangle$.}
\end{figure}

The $\operatorname{\textsf{find}}_{\bm f}(x)$ operation is now simply a standard
search on the sideways heap: we start at the root $\lambda n$,
thus knowing the sum of the first $2^{\lambda n}$ values, and then move to the
left or right child depending on whether $x$ is smaller or greater than that sum. Note that when we move to the
left child $x$ remains unchanged, but when we move to the right child we have to
subtract from $x$ the partial sum stored in the current node. Non-existing nodes
are simply ignored (we follow their left child; see line~\ref{line:ignore} of
Figure~\ref{fig:find}).

Fenwick shows that $\operatorname{\textsf{find}}_{\bm f}(x)=\operatorname{\textsf{find}}_{\bm v}(x)$.
Figure~\ref{fig:find} illustrates with an example the sequences of nodes traversed
by a \textsf{find} operation.

\subsection{The interrogation tree}

To compute a prefix sum on $\bm v$ we use the \emph{interrogation tree}. The
tree has an implicit root $0$. The parent of node $j$ is
\[\operatorname{\textsf{int}}(j) = j - (1 \mathbin\ll \rho j) = j \mathbin\&
(j-1),\] that is, $j$ with the lowest one cleared; note that the resulting tree
is not a binary tree (in fact, the root has infinite degree).
The primitive $\operatorname{\textsf{prefix}}_{\bm f}(j)$ accumulates partial sums on the path from node
$j$ to the root (excluded), and Fenwick shows that
$\operatorname{\textsf{prefix}}_{\bm f}(p) = \operatorname{\textsf{prefix}}_{\bm v}(p)$.
Figure~\ref{fig:int} illustrates with two examples the sequences of nodes traversed
by a \textsf{prefix} operation.

Note that we use the
same name for the abstract operation on the list and the implementation in a
Fenwick tree to avoid excessive proliferation on names. The subscript should
always make the distinction clear.

\begin{figure}
\centering
\begin{tabular}{cc}
\begin{minipage}[b]{7cm}
  \begin{algorithmic}[1]
    \Function{$\mathrm{prefix}$}{$j$}
      \State $p \gets 0$
      \While{$j \neq 0$}
        \State $p \gets p + f[j]$
        \State $j \gets j \mathbin\& (j-1)$
      \EndWhile
      \State \Return{$p$}
    \EndFunction
  \end{algorithmic}
  \vspace*{1cm}
\end{minipage}&
\includegraphics{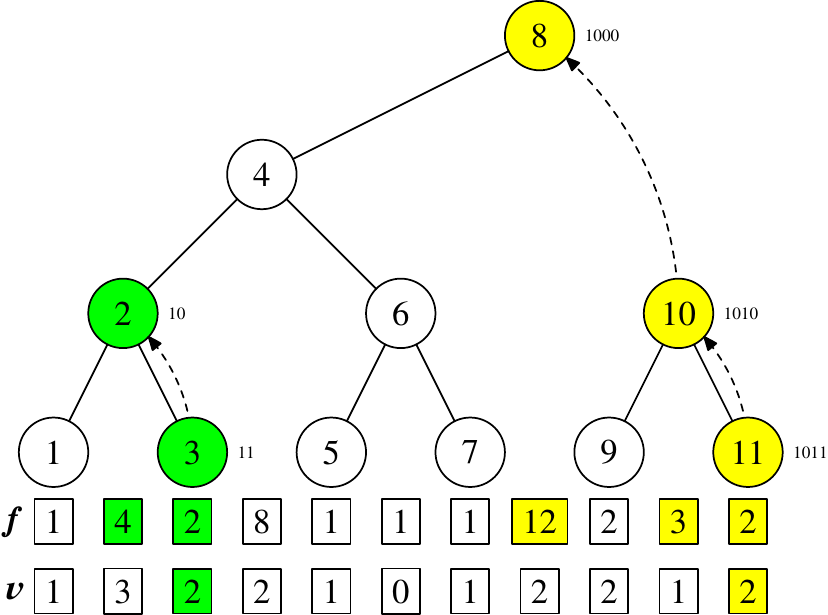}
\end{tabular}
\caption{\label{fig:int}Pseudocode for \textsf{prefix}, and two sample executions on the tree
of Figure~\ref{fig:tree}: on the left the nodes traversed computing the
sum of the first three elements, and on the right the nodes traversed computing
the sum of the first eleven (i.e., all) elements. The value returned is the
sum of the elements of $\bm f$ associated with the visited nodes. Note how the sequence of
node indices is obtained by deleting the lowest bit set.}
\end{figure}

We remark that the interrogation tree can be also scanned in
a \emph{top-down} fashion, that is, from the root to the leaves. The
sequence of values of $j$ that reaches the leaf $p$ starting from the root
(i.e., $j = 0$) is given by the rule \[j \leftarrow j \mathbin| \bigl(1\ll\lambda(j
\oplus p)\bigr).\] The
rule adds to $j$ the bits sets in $p$ one by one, from the most significant one
to the least significant one, thus selecting at each step the correct child.

\subsection{The update tree}

Finally, when modifying $v_j$ by adding the value $c$, Fenwick shows that we
have to update $\bm f$ by adding $c$ to all nodes along a path going up from
node $j$ in the \emph{update tree}.
In this tree, the parent of node $j$ is \[\operatorname{\textsf{upd}}(j) = j +
(1 \mathbin\ll \rho j) = j + (j \mathbin\& -j).\] If $j$ is an odd multiple of
$2^a$, that is, $j=c2^a$ with $c$ odd, its parent is $(c+1)2^a$.  We stop
updating when the next node to update is beyond $n$. We call this operation
$\operatorname{\textsf{add}}_{\bm f}(j, x)$. Figure~\ref{fig:upd} illustrates
with two examples the sequences of nodes traversed by a \textsf{add} operation.

\begin{figure}
\centering
\begin{tabular}{cc}
\begin{minipage}[b]{7cm}
  \begin{algorithmic}[1]
    \Function{$\mathrm{add}$}{$c,j$}
      \While{$j \le n$}
        \State $f[j] \gets f[j] + c$
        \State $j \gets j + (j \mathbin\& -j)$
      \EndWhile
    \EndFunction
  \end{algorithmic}
  \vspace*{1cm}
\end{minipage}&
\includegraphics{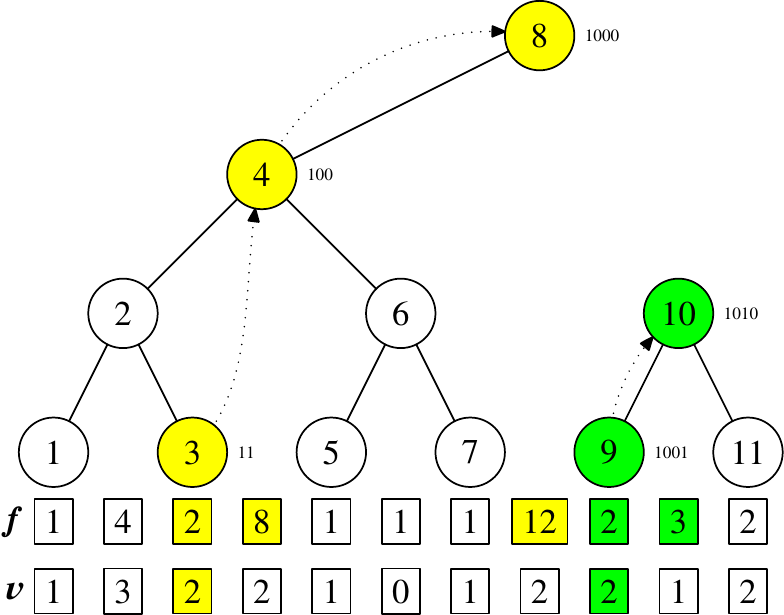}
\end{tabular}
\caption{\label{fig:upd}Pseudocode for \textsf{add}, and two sample executions on the tree
of Figure~\ref{fig:tree}: on the left the nodes traversed
when updating the third element, and on the right the nodes traversed updating
the ninth element. During the traversal, the argument $c$ is added to all elements
of $\bm f$ associated with the visited nodes. Note how the sequence of node indices is obtained by adding the lowest bit set.}
\end{figure}

Note that technically the update tree is a forest, but for convenience and
historical reasons we still refer to it as a tree. More precisely, as in the
case of sideways heaps we could add implicit nodes $(n\.. 1 \ll (\lambda n +
1)]$ containing no partial sum which would make the forest into a tree.

It is easy to prove the following remarkable duality:
\begin{proposition}
$\operatorname{\textsf{int}}(-j) = -\operatorname{\textsf{upd}}(j)$.
\end{proposition}
\begin{proof}
Immediate, as $\rho j = \rho (-j)$ and $\operatorname{\textsf{int}}(-j) =(-j) - (1\ll \rho (-j))=
-(j + (1\ll \rho j)) = -\operatorname{\textsf{upd}}(j)$.
\end{proof}

The equation above can be used to provide a top-down scanning towards the leaf
$p$, analogously to the case of \textsf{int}.
The rule is simply \[j \leftarrow  -\bigl(-j \oplus
(1\ll\lambda(-j \oplus -p))\bigr),\] starting from $j=n\mathbin\& \bigl(-1 \ll
\lambda\bigl(n\oplus p \oplus (1 \ll \rho p)\bigr)$, which in two's complement
arithmetic can be written as $n\mathbin\& (-1 \ll
\lambda\bigl(n\oplus (p \mathbin\& (p - 1))\bigr)$.
Alternatively, using the convention that $\lambda0=-1$ and that
negative shifts are shifts in the opposite direction one can
express the starting node as $j=n\mathbin\& \bigl(-1 \ll
\lambda(n\oplus p)\bigr)$.

While apparently the top-down rule for \textsf{upd} requires many more
operations, we can in fact perform the whole scan using negative indices: once
the initial node has been set up in this way, the only difference between the
two algorithms is that we have to remember to negate the result of the update
rule to obtain the actual node of the tree.

\subsection{Bounded lists and complemented find}
\label{sec:compfind}

Departing from the original definition, we will also assume the existence of an
upper bound $B$ on the values in the list represented by the tree (i.e.,
$v_i\leq B$).
Note that this is a not bound on the partial sums, as
they are sums of such values, but it is easy to see that the upper bound for the
partial sum of index $j$ will be $2^{\rho j}B$, so one needs
$\bigl\lceil\lg\bigl(2^{\rho j}B+1\bigr)\bigr\rceil$ bits
to store it.

Using the bound $B$, we now define an additional operation, \emph{complemented
find}, by
\[
\findc_{\bm v}(x) = \max \bigl\{\? p \mid pB - \prefix_{\bm v}(p) \le
x\?\bigr\} \qquad x\in\N.
\]
This operation is not relevant when using a Fenwick tree to keep track of prefix
sums, but will be essential to implement selection on zeros.

An implementation of complemented find is given in~\cref{findc}.
Note that the only difference with a standard find (Figure~\ref{fig:find}) happens at \cref{comp_val}:
instead of considering the value of a node, we compute the difference
with the maximum possible value stored at the node.

For both versions of find, we actually consider \emph{two} return values:
the length of the longest prefix with (complemented) sum less than or equal $x$,
\emph{and} the excess with respect to said (complemented) prefix sum. Once again,
in the case the Fenwick tree is used just to keep track of prefix sums the excess
is not particularly useful, but it will be fundamental in implementing
selection primitives.

\begin{algorithm}
  \caption{Complemented find for a Fenwick tree $\bm f$ with bound $B$.}
  \label{findc} \begin{algorithmic}[1]
    \Function{$\overline{\mathrm{find}}$}{$k$}
      \State $p \gets 0$
      \State $q \gets 1 \mathbin\ll \lambda n$
      \While{$q \neq 0$}
        \If{$p + q \leq n$}
        \label{line:ignore}
          \State $m \gets B2^{\rho(p+q)} - f[p+q]$
          \label{comp_val}
          \If{$k \ge m$}
            \State $p \gets p + q$
            \State $k \gets k - m$
          \EndIf
        \EndIf
        \State $q \gets q \mathbin\gg 1$
      \EndWhile
      \State \Return{$\langle p, k\rangle$}
    \EndFunction
  \end{algorithmic}
\end{algorithm}

\section{Compression and layout}

Recent computer hardware features a growing performance gap
between processor and memory speed. Keeping most of the computation into the
cache can increase significantly the performance of data structure.
Thus, \emph{compressed}, \emph{compact} or \emph{succinct} data structures
may require more complex operation for queries and updates, but, at the
same time, by storing information more compactly they make it possible to
increase performance by using the cache more efficiently.

We have explored several different forms of compression in order to balance the
trade-off between the time required to retrieve uncompressed information and the
time required to access to such data from the physical memory. After an
experimental analysis we are reporting here the two most relevant ones.

\subsection{Bit compression}
\label{sec:bitcomp}

As we explained in Section~\ref{sec:tree}, we assume to have a bound $B$ on the
values $v_i$ represented by the tree; let $S=\lceil\lg(B+1)\rceil$.
In our implementation, we use $S$
bits for the leaves, $S+ 1$ bits for the nodes at level
one, and so on. In general, we use $S+\rho j$ bits to store the
partial sum at node $j$ (which is slightly redundant, in particular when $B$ is a power of two).
\begin{figure}
\centering
\includegraphics{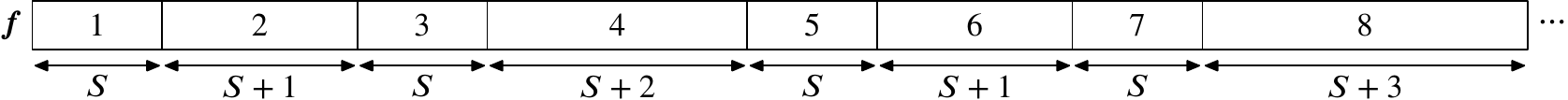}
\caption{\label{fig:bit}The bit-compressed Fenwick layout. Each box contains the
index $j$ of the associated node, and the width in bits of each box is $S+\rho j$.}
\end{figure}

To be able to fully exploit bit compression, we need to lay out all partial sums
consecutively in a bit array: an example is shown in Figure~\ref{fig:bit}. A useful fact is that if we assume we can compute in
constant time sideways sums, then we can compute in constant time the sum of the
bit sizes of the first $j$ partial sums.

\begin{proposition}
\label{prop:formula}
For $0\leq j\leq n$ we have
\[
\sum_{i=1}^{j} \bigl(S+\rho i ) = j\cdot (S +1) - \nu j.
\]
\end{proposition}
\begin{proof}
The only part that needs proof is $\sum_{i=1}^{j} \rho i = j - \nu j$, which is
noted by Knuth.~\cite[7.1.3]{KnuACPIV} 
\end{proof}

Note that the number
of possible configurations of the tree is $(B+1)^n$, so the
information-theoretical lower bound to store the tree is $\lceil n\lg
(B+1)\rceil$ bits.
Adding up on all nodes, similarly to Proposition~\ref{prop:formula} we obtain
\[ \sum_{i=1}^n(S + \rho i)
= n S + n - \nu n
< n (1 + S)
= n (1 + \lceil\lg(B+1)\rceil)\leq
2n \lg(B+1)\leq
2\lceil n\lg(B+1)\rceil .\]
Thus, our bit-compressed Fenwick tree
qualifies as a \emph{compact} data structure---one that uses space that is a
constant number of times the information-theoretical lower bound.

For $B=1$ and $n$ a power of two our bound is
tight, but for larger values of $B$ it is rather rough: for example, when $B=63$ we use less than $7n$ bits, whereas the information-theoretical
lower bound is $6n$. In fact, for $B>3$ a better bound is
\[\sum_{i=1}^n(S + \rho i)\leq n (1 + \lceil\lg(B+1)\rceil)\leq 2n +\lceil n\lg(B+1)\rceil,\]
which shows that we are actually losing at most two bits per element with respect to the lower bound; this bound
is tight when $B$ is a sufficiently large power of $2$.
%

\subsubsection{Byte compression}
\label{sec:bytecomp}

Another possibility, trading (in principle) space for speed, is to round up the
size of each node to its closer byte. In this setting, the partial sum for node $j$ requires
$\bigl\lceil \bigl(\rho j + S\bigr)/8\bigr\rceil$ bytes and starts at byte
\[
\sum_{k=1}^{w/8} i \left\lceil \frac{S}{8k} \right\rceil.
\]
The main problem of this approach is that we cannot provide a closed-form
formula like that given in Proposition~\ref{prop:formula}. To compute the
starting byte of node $j$ we need to establish how many nodes using at least $k$
bytes appear before $j$.

To overcome this issue, we suggest an even looser compression strategy:
instead of rounding up the space required by each node to the next byte, one can
use just three byte sizes for partial sums: $\left\lceil (\lg (B+1)) / 8
\right\rceil$, $\left\lceil (\lg (B+1)) / 8 \right\rceil + 1$ or full size (no
compression). Let $b = \left\lceil (\lg (B+1)) / 8 \right\rceil$ and let
$d = 8b - \lceil \lg (B+1) \rceil$; the sum of the byte size of the first $j$
partial sums is
\begin{equation} \label{eq:formula_byte}
\overbrace{j b}^\text{at least $b$ bytes for each node} +
\overbrace{(j \mathbin\gg d)}^\text{nodes with at least one additional byte} +
\overbrace{\left(j \mathbin\gg (d + 8)\right) \cdot (w/8 - b + 1).}^\text{uncompressed nodes multiplied by their additional space}
\end{equation}
In this way, one provides a byte-sized optimal compression for the lower nodes
of the tree, and leave uncompressed everything above. This simplification has
practically no impact on the overall space used by the Fenwick tree, as just a
very small fraction of nodes will be represented in uncompressed form.
For example, even in trees with height $64$ (the worst-case scenario if you
want to store each node in a single word of a $64$-bit processor) only about the
$0.026\%$ of the nodes reside in the third (uncompressed) partial sum.

As we will see in our experimental analysis, this looser byte compression
mechanism happens to offer often faster access to the partial sums than the
bit compression scheme.

\subsection{The level-order layout}
\label{sec:level}

Compression is not the only way to improve cache efficiency. Another
important aspect is \emph{prefetching}: instead of fetching the data when it is
needed, a prefetcher works by guessing what the next requested data will be and
fetching it in advance. If the prefetching succeeds---that is, if the fetched
data will be actually used in the near future---a cache miss is prevented and
the execution will be faster. If the prefetcher makes the wrong guess,
however, a cache line will be replaced with useless information and the
prefetching might cause additional cache misses. Both hardware and software
prefetchers exists, and the key to take advantage of their capabilities is to
dispose the data in a way that the accesses made by the most frequent
operations follow a simple and well defined pattern.

To find a predecessor, the Fenwick tree needs in the worst case (and also in the
average case) to query the value of a logarithmic number of nodes.
At each step, one of the children of the current node will be selected to
continue the search, and the distance between two children of a node of height
$h$ is $2^{h-1}$: one of the children, and we cannot predict which, will be
reached in the next iteration. For $h$ enough large so that the two children lie
in different cache lines the prefetcher can either make a guess and succeed with
a probability of $1/2$, or it can prefetch both of them, wasting a cache line
and doubling the required number of memory accesses.

To help cache prefetchers in making the correct guess every time we might
modify the layout of the partial sums, and proceed in \emph{level order}:
first the root, then the children of the root, and so on. At
that point, each pair of children would be at distance one, and thus
much likely on the same cache line.\footnote{Bille~\textit{et al.}~\cite{BCPSPSFT} have proposed the same layout
with completely different motivations.}

Of course this change of disposition does not come for free.
If we lay out the levels in a single array, with pointers locating the start of
each level, it will become very difficult to increase the number of nodes of the
tree: this solution is thus restricted to Fenwick trees with immutable size.
The alternative is to add a level of indirection, and thus to have an array of
pointers, each to a different level. As a small advantage, since each level is
compressed in the same way, once we locate the start of a level it is easier to
find a node than in the case of a Fenwick layout.

We will use always such an
implementation; as a consequence, every node in the level-order layout will be
identified by two integers: its \emph{level} $\ell$ and its \emph{zero-based
index} $k$ in its level. Nodes with $\ell=0$ are the leaves.
This representation and the classic Fenwick representation are connected by the
following easily proved proposition:
\begin{proposition}
\label{prop:corr}
A node with index $j$ in the classical Fenwick layout has level $\rho j$ and
index $j \mathbin\gg (1+\rho j)$ in the level-order layout. A node with level $\ell$ and
index $k$ in the level-order layout has index $(2k + 1)\ll \ell$ in the
classical Fenwick layout.
\end{proposition}

This bijection induces parent-children relationship on the level-order
layout: in particular, it is easy to show that:
\begin{itemize}
  \item for $\ell>0$, the children of node $\langle \ell, k\rangle$ in the
  sideways heap are $\langle \ell-1, 2k\rangle$ and $\langle \ell-1,
  2k+1\rangle$ (but note that the second child might not exist);
  \item for $0 < k \le n \mathbin\gg (\ell + 1)$, the parent of node $\langle \ell,
  k\rangle$ in the interrogation tree is $\langle \ell + 1 +\rho k , k \mathbin\gg (1 + \rho k )\rangle$;
  \item for $0\leq k < n \mathbin\gg (\ell + 1)$, the parent of node $\langle \ell,
  k\rangle$ in the update tree is $\bigl\langle \ell + 1 + \rho \bar k , k \mathbin\gg \bigl( 1+ \rho \bar k
  \bigr)\bigr\rangle$.
\end{itemize}

\section{Alignment problems}
\label{sec:alignment}

The Fenwick tree in Fenwick layout is extremely sensitive
to the alignment of the partial sums. In our experiments,
alignment problems can cause a severe underperformance of the primitives, which
can be up to three times slower in the case of the bit-compressed tree.

\subsection{Misalignment}

To understand why this can happen, consider a standard technique for reading
a sequence of bits which starts at position $p$ when $p$ is not byte-aligned (i.e., a multiple of eight):
one uses an unaligned word read at byte position $\lfloor p / 8 \rfloor$,
and then one suitably shifts and masks the result (note, however, that this method does not allow in general
to read bit blocks of length close to $w$, as they might span multiple words).

Consider now a bit-compressed tree of $n$ nodes, with $n$ larger than the size $M$ of a memory
page in bits, which we assume to be a power of two. As it is easy to
see from the nature of the parent-child relationship of the three implicit trees
described in Section~\ref{sec:tree}, in all operations the
partial sums whose indices have a large number of trailing zeros in binary representation
(i.e., $\rho j$ is large, where $j$ is the index) are the ones more frequently accessed.

Let as assume that the memory containing the partial sums is aligned with a memory page (e.g., if \verb|mmap()| is used to allocate memory), and
let us look at the bit position in memory of the partial sum of index $1\ll\lambda n$.
Using Proposition~\ref{prop:formula}
it is immediate to see that it will be
\[
(1\ll\lambda n)\cdot(S+1) - 1 - S - \rho (1\ll\lambda n).
\]
Thus, if $S+1+\rho(1\ll\lambda n) =S+1+\lambda n\leq w - 8$, reading the partial sum as above will cause a word read \emph{across} memory pages, which
will access to \emph{two} distinct memory pages and
usually generate \emph{two} cache
misses.
\begin{figure}
\centering
\includegraphics{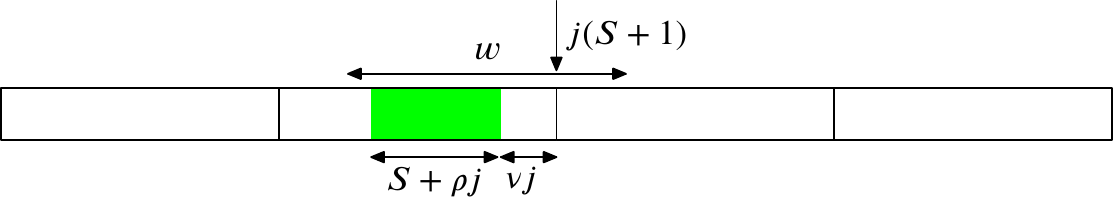}
\caption{\label{fig:mis}An example of misalignment: rectangles are word-aligned words, and $j(S+1)$ is on the border of a memory page. By Proposition~\ref{prop:formula} the
partial sum of node $j$ is stored in the highlighted block of $S+\rho j$ bits starting at $j(S+1)-\nu j- (S+\rho j)$, but using unaligned reads we would actually access the $w$ bits indicated
by the double arrow, thus accessing two distinct memory pages. Due to the 
layout of the Fenwick tree, these kind of reads are very frequent.}
\end{figure}

As shown in Figure~\ref{fig:mis}, the same argument applies to any node whose index $j$ is such that $2^{\rho j}$ is a multiple of $M$
(so the partial sum is stored at the very end of a memory page)
and  $1 + S + \rho j\leq w - 8$ (so we read an across-page word). As we mentioned, this high number of across-page
read/write operation can slow down the data structure by a factor of three, because
it affects the nodes that are accessed more frequently.\footnote{The reader might be tempted to suggest to read partial sums \emph{backwards}, that is,
performing an unaligned access using as reference the \emph{last} byte occupied by the partial sum, but this would simply make the problem appear when the
tree is in a different memory position, that is, a word after the start of a memory page.}

To work around this problem, we note the following property:

\begin{proposition}
\label{prop:aligned}
If $j\cdot (S+1)$ is a multiple of $w$, and the memory representation
of the tree is shifted by one bit to the right, the partial sum of index $j$ is contained in a word-aligned
word.
\end{proposition}

\begin{proof}
Due to the one-bit shift, the partial sum of index $j$ will start at position $j\cdot (S+1) - \nu j - (S+\rho j) + 1$ (included)
and end at position $j\cdot (S+1) - \nu j +1$ (excluded). However,
\[
\nu j + \rho j \leq \lambda j + 1 \quad\text{for $1\leq j\leq n$,}
\]
because $\nu j$ is the number of ones in the binary representation of $j$, whereas $\rho j$ is bounded by the number of zeros, so
\[
\nu j  - 1 +  S + \rho j \leq S + \lambda j \leq w,
\]
because the largest size in bits of a partial sum cannot exceed $w$. Thus, the partial sum
of index $j$ is contained in the word of (zero-based) index $j\cdot (S+1)/w - 1$.
\end{proof}
As a consequence, we can read frequently accessed partial sums (which have a large $\rho j$, so $j\cdot (S+1)$ is a multiple of $w$)
 using word-aligned words (i.e., if the partial sum
starts at position $p$, we read the word of index $\lfloor p/w\rfloor$ and then suitably shift and mask).
The proposition will work also when $j\cdot(S+1)$ is a multiple of eight, albeit the word containing the partial sum
will be just byte-aligned.
When $j\cdot(S+1)$ is \emph{not} a multiple of eight, if we want to guarantee that a single word access
will retrieve correctly the partial sum we have to make some assumption: in
particular, the largest partial sums we might read, which will be made of $S+2$ bits (as $\rho j < 3$),
must be at most $w-7$ bits long, that is, $S\leq w - 9$.
In practice,
on $64$-bit architectures we have to require $S\leq 55$, that is, $B < 2^{55}-1$,
which is more than sufficient for our purposes. Alternatively, one has to resort to read one or two
word-aligned words and reconstruct the partial sum from those.

Note that this bound is tight: if $w=64$,
$S=56$ and $j=60$ the last bit of the associated $58$-bit partial sum will be in position $3480$, so
the partial sum will span two words.

We remark that the byte-compressed Fenwick tree in Fenwick layout has no
across-page access problems, as all partial sums are byte-aligned, and
if $\rho j$ is large enough all three summands of the
offset formula~(\ref{eq:formula_byte}) are multiples of $w$.

\subsection{Hyper-alignment}

The second alignment problem, which is more subtle and affects even
the classical version, is that frequently accessed partial sums are stored at memory locations
with the same residual modulo $2^aw$, for small $a$. This is obvious in the classical construction, as frequently
accessed partial sums have indices with a large $\rho j$. In
the bit-compressed tree, instead, this pattern is due to access being modeled using Proposition~\ref{prop:aligned}.
This kind of \emph{hyper-aligned} access implies that all the nodes end up being cached in the same way of the cache,
as multi-way caches usually decide in which way to store a cache line using the lowest bits of the address.

The solution we adopt is that of perturbing slightly the data structure so to spread the residuals of the memory addresses
of frequently accessed partial sums across the whole possible range of values. To obtain this result,
we insert $w$-bit sized \emph{holes} at regular intervals. By making the intervals a sufficiently large power of two,
the impact on space usage is negligible\footnote{In our code, we use an interval of $2^{14}$, so the space usage
increases by less than one per thousand, and the partial sum of index $j$ is stored in position
$j+(j\gg 14)$. In the case of bit compression, we offset the starting bit by $64\cdot(j \gg14)$.}, and computing the actual position requires just an additional shift
and a sum. By making the holes $w$-bit wide, all the good alignment properties of the structure (i.e., Proposition~\ref{prop:aligned})
are preserved. The find operations are the one more affected: a standard Fenwick tree on a large list (i.e., 100M elements)
features 30\% shorter execution times on our hardware once holes are put in place.

\section{Ranking and selection}
\label{sec:rankselect}

Let us now get back to problem of dynamic rank and selection. Consider a
\emph{dynamic} bit vector $\bm b$: that is, we are able to change the value of
$b_i$ (\emph{set} and \emph{clear} operations) and expand or reduce the length
of $\bm b$ ($\push$ and $\pop$ operations).

If the vector is very short, it is possible to perform ranking and selection
very quickly using \emph{broadword programming}. Modern computers are capable of
computing sideways additions on a word using a single instruction,
usually called \emph{popcount} (population count). With this instruction we can
easily perform ranking within a word after suitably masking. Also selection in a word
can be performed quickly,\cite{VigBIRSQ,GoPOSDSMD} particular on very
recent Intel architectures which provide the PDEP instruction. In practice, since
sequential memory accesses are highly amenable to cache prefetching, a linear
search is still the fastest way to perform dynamic ranking
and selection on a sufficiently small bit vector.

For larger vectors, we can split the problem into ranking and selection over
\emph{blocks} of $q$ bits, where a linear scan is
sufficiently fast, and keeping track of the prefix sums of the number of ones in each block
using a Fenwick tree $\bm f$. At that point, we can implement rank and select as
\begin{align*}
\rank_{\bm b}(p) &= \prefix_{\bm f}(m) + \rank_{\bm b[mq\..
(m+1)q)}(p \bmod q)& &\text{where $m=\lfloor p/q\rfloor$}\\
\select_{\bm b}(k) &= pq + \select_{\bm b[m q \..
(m+ 1) q)} (g)& &\text{where $\langle p,g\rangle = \find_{\bm f}(k)$ and
$m=\lfloor p/q\rfloor$}\\
\end{align*}

Note that this strategy allows us to compute rank and select queries on zeros
too. Select on zeroes consists in replacing $\find$ with $\findc$ and flipping
all the bits during linear scans: the bound $B$ of Section~\ref{sec:compfind} is
$q$. Rank on zeros is simply $p - \rank_{\bm b}(p)$.

We have previously assumed we are capable of computing ranking and
selection within a word using some extremely fast operations, such as
dedicated assembly instruction. This assumption suggests the
choice of a fairly large $q$, as the bigger is $q$, the smaller will be
additional space required to store the Fenwick tree. However, many RISC
architectures such as the widespread ARM are not capable of performing these
operations in few clock cycles. In this case, the compact structures described in this
paper might be even more relevant, as a small $q$ implies a larger tree.

\section{Experiments}

In this section we present the results of our experiments, which were performed
on an Intel\textregistered{} Core\texttrademark{} i7-7770 CPU @3.60GHz (Kaby
Lake), with 64\,GiB of RAM, Linux 4.17.19 and the GNU C compiler 8.1.1.
We use \emph{transparent huge pages}, a relatively new feature of the operating system
that makes available virtual memory pages of $2$\,MiB, which significantly reduce
the cost of accessing the TLB (Translation Lookaside Buffer) for large-scale data structures.


We will present two sets of experiments:
\begin{itemize}
  \item In the first set, we examine the primitives of the Fenwick tree in
  isolation, exploring the effects of layout and compression policies on speed. We generate
  list of values chosen uniformly at random, and average the running time of a large
  number of queries at random locations. Note that both $\prefix$ and
  $\add$ are \emph{data-agnostic}, in the sense that the execution flow
  does not depend on the content of the tree; the execution flow of $\find$, on the contrary, depends
  on the content.
  \item In the second set, we use a Fenwick tree to support rank and
  selection: in this case, beside the way we structure the tree
  we also explore the effect of different block sizes; as a baseline, we report results for Nicola Prezza's library
	for dynamic bit vectors~\cite{PreFDDSSP}, which is the only practical available
	implementation we are aware of.\footnote{We remark that Prezza's implementation provides additional primitives for \emph{bit insertion} and \emph{bit deletion}, which
is not available through the Fenwick tree.} We generate random bit vectors with
  approximately the same number of ones and zeroes, and, again, average the running time
  of a large number of queries at random locations.
\end{itemize}
In both cases, we take care of potential \emph{dead-code elimination} by assigning the
value returned by each test to a variable which is in the end assigned a volatile dummy variable.
Moreover, to avoid excessive and unrealistic speculative execution, in a sequence of calls
to a primitive we xor the next argument using the lowest bit of the previously returned value,
as this approach creates a chain of dependencies among successive calls.

We do not report the results for
complemented find and since both its algorithm and its memory-access pattern are
very similar to that of a standard find, so their performances are very similar. The same argument
holds for ranking and selection on zeros.

Finally, we report results in two real-world applications: counting the number 
of transpositions generating a permutation, and generating pseudorandom graphs
following a \emph{preferential attachment} model.

\subsection{Variants}

\begin{table}
\centering
\begin{tabular}{lr}
Name & Comment\\
\hline
bit$[\ell]$& Level-order layout; bit compression\\
byte$[\ell]$& Level-order layout; byte compression\\
fixed$[\ell]$& Level-order layout; fixed-width representation\\
bit$[F]$&Classical Fenwick layout; bit compression\\
byte$[F]$&Classical Fenwick layout; byte compression\\
fixed$[F]$&Classical Fenwick layout; fixed-width representation\\
\end{tabular}
\caption{\label{tab:names}Naming scheme for the Fenwick trees appearing in
experiments. Implementations of dynamic bit vectors
have an additional number denoting the size of a block in word (i.e.,
bits$[\ell]8$).}
\end{table}

We consider a number of implementations displayed in Table~\ref{tab:names}. We
use three different kinds of compression strategies, namely: no compression
(dedicating $64$~bits for each node), byte compression (see Section~\ref{sec:bytecomp})
and bit compression
(see Section~\ref{sec:bitcomp}). For each of them we consider the classical and level-order layout.\footnote{Even
though in Section~\ref{sec:level} we provided parent formulas for the update
and interrogation trees in terms of levels, we found that is faster to use the standard
formulas for the Fenwick layout and convert the result using Proposition~\ref{prop:corr} at each iteration.}
Since we want to use the results of this benchmark as a reference for choosing
good parameters for the rank and selection structures, we are fixing to $64$ the
bound $B$ on the leaves.

We remark that throughout our graphs a few spikes are visible: they are due to
the nonlinear behavior of the tree, as when the size reaches a power of two a
new root is created, usually in a new memory page if the tree is page-aligned,
and the tree becomes as unbalanced as possible (e.g., updating an element will
require one additional iteration). This fact can increase the number of
TLB/cache misses and the number of average accesses: as the number nodes grow,
this cost is amortized.

\subsection{Fenwick trees}

Figures~\ref{fig:bench_prefix}-\ref{fig:bench_add} report the performance of several variants
for the $\prefix$, $\find$ and $\add$ operations. Note that we do not report data about $\push$/$\pop$, but a $\pop$ primitive
would be constant-time (just update the size of the tree, as indices of parents in the update tree are always larger than indices of their children),
while the performance of the $\push$ primitive is essentially identical to that of $\add$ (we will, however, give a sample application using $\push$).
\begin{itemize}
  \item In the performance of $\prefix$, it is immediate evident the point in which each structure
  goes out of the L3 cache. Non-compressed structures have a slight advantage initially, but they are
  almost twice as slow when the compressed structures still fits in cache. After all structure go out
  of the cache, the advantage of the simpler code of non-compressed structure shows again, albeit the
  speedup is very marginal compared to the byte-compressed structures. The bit-compressed structure
  have significantly more complex code, and in particular a test that decides whether to align
  access to bytes or words, and this cost increases logarithmically.
  \item The graphs of $\find$ are somehow smoother, due to the more regular cache usage. What is
  evident here is the bad performance of the classical Fenwick tree, and that the fastest
  structure (except for very small scale) is the level-order byte-compressed version. Note that
  a $\find$ operation requires always $\lg n$ accesses, whereas in the other cases, as noted by
  Fenwick~\cite{FenNDSCFT}, assuming a perfect sideways heap only $(1 + \lg n) /
  2$ are required on average.
  \item The graphs for $\add$ show a behavior similar to that of $\prefix$, but with less pronounced
  differences in performance. In particular, the Fenwick and the level-order performances are much closer.
\end{itemize}

Note that the advantage of compressed trees in underestimated by our benchmarks, as the cache will be shared with
other parts of the code. Benchmarking different kind of trees inside a specific application is the most reliable way
to determine which variant is the most efficient for the application at hand.

\begin{figure}
\centering
\begin{minipage}{.5\textwidth}
  \centering
  \includegraphics[width=\linewidth]{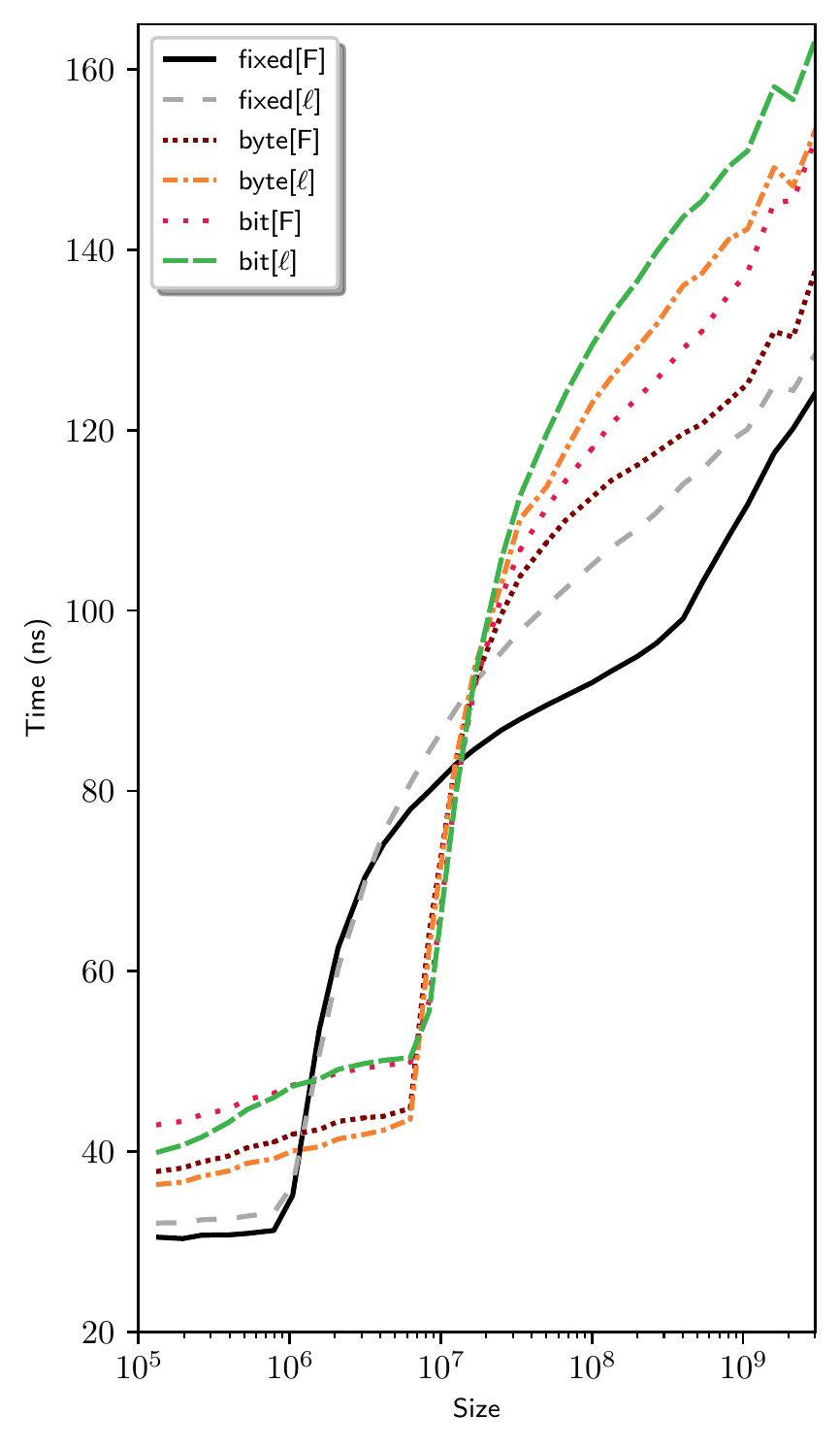}
  \captionof{figure}{\label{fig:bench_prefix}Performance of $\prefix$.}
\end{minipage}%
\begin{minipage}{.5\textwidth}
  \centering
  \includegraphics[width=\linewidth]{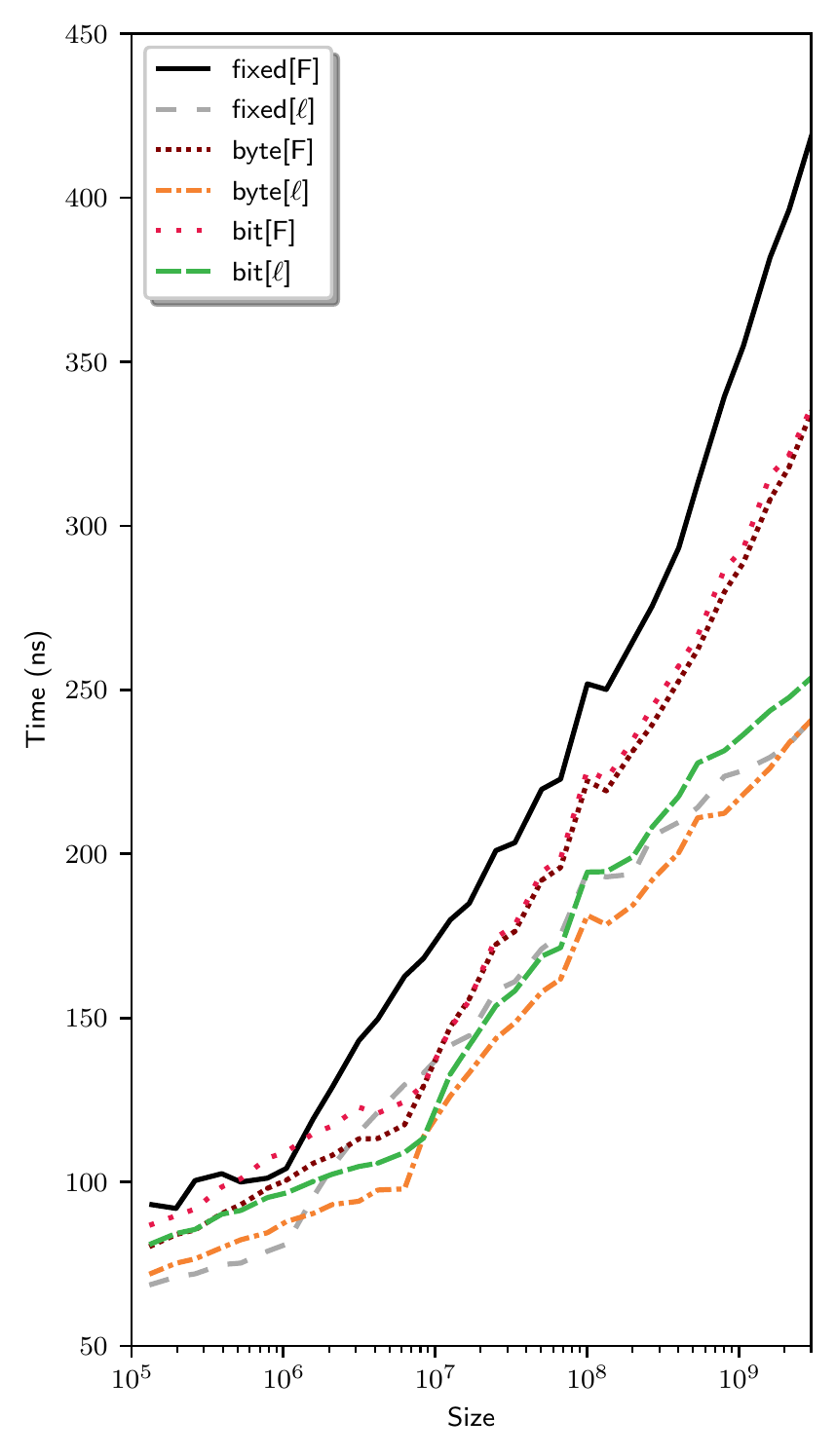}
  \captionof{figure}{\label{fig:bench_find}Performance of $\find$.}
\end{minipage}
\end{figure}

\begin{figure}
\centering
\begin{minipage}{.5\textwidth}
  \centering
  \includegraphics[width=\linewidth]{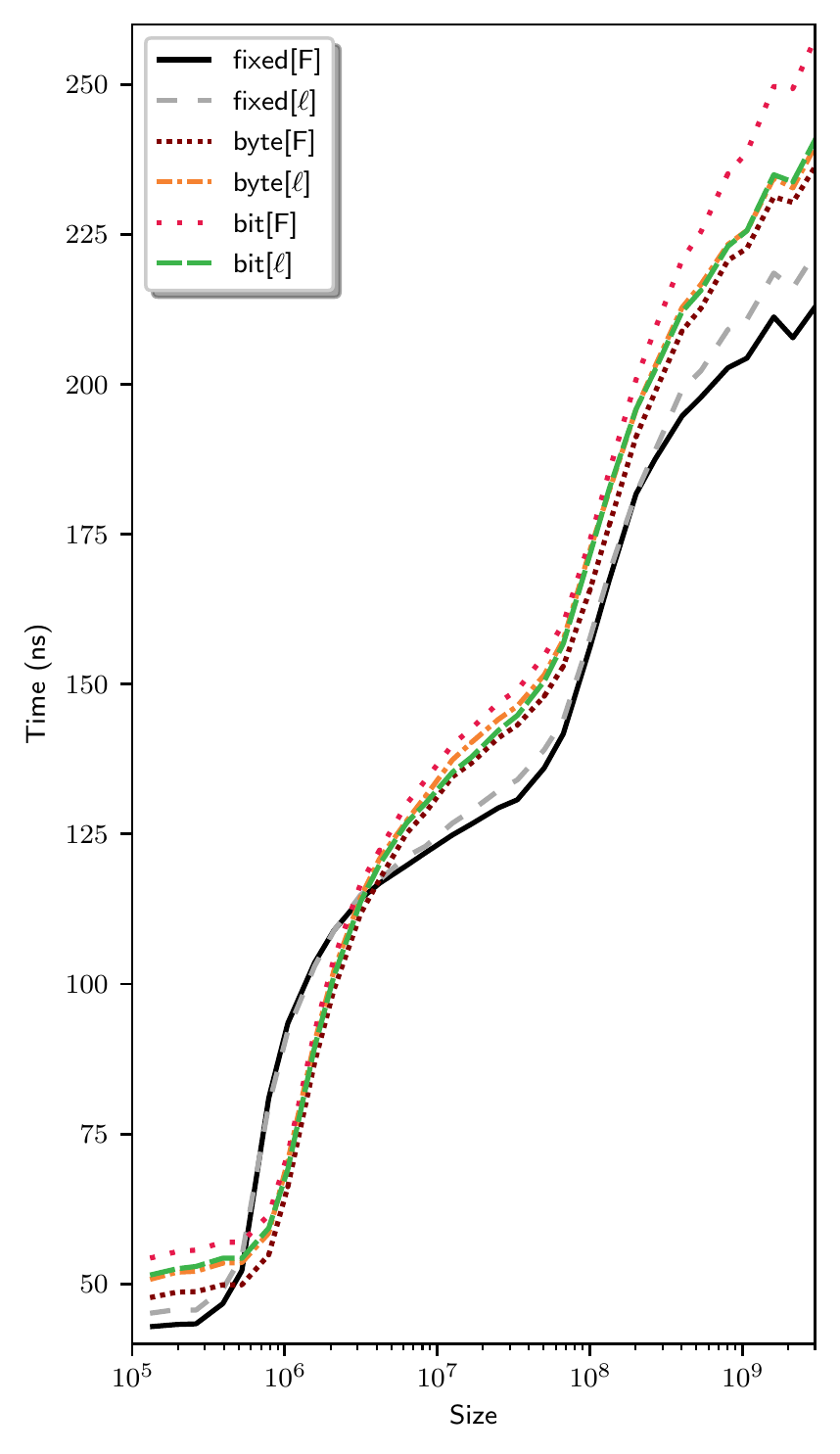}
  \captionof{figure}{\label{fig:bench_add}Performance of $\add$.}
\end{minipage}%
\begin{minipage}{.5\textwidth}
  \centering
  \includegraphics[width=\linewidth]{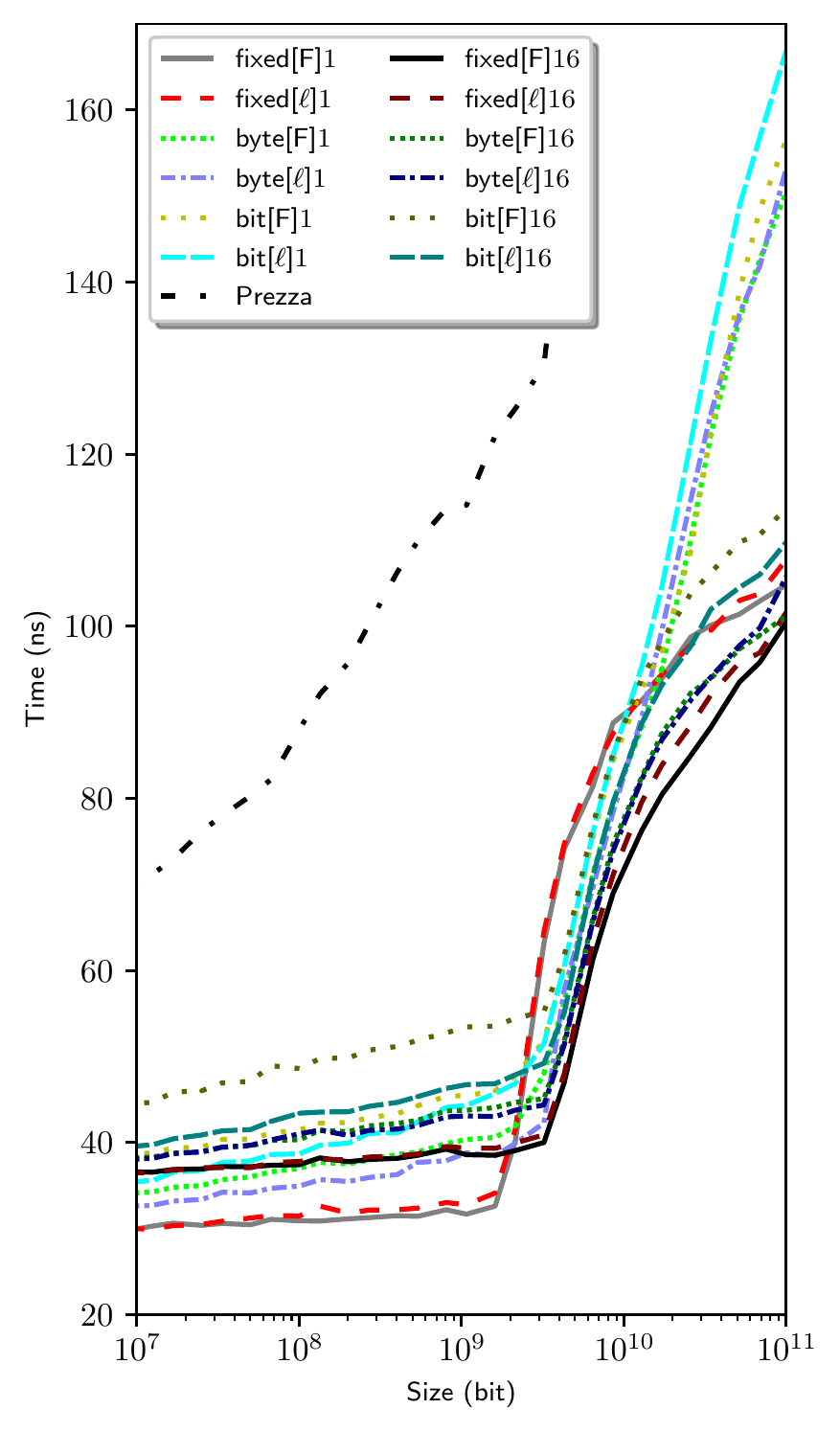}
  \captionof{figure}{\label{fig:bench_rank}Performance of $\rank$.}
\end{minipage}
\end{figure}

\begin{figure}
\centering
\begin{minipage}{.5\textwidth}
  \centering
  \includegraphics[width=\linewidth]{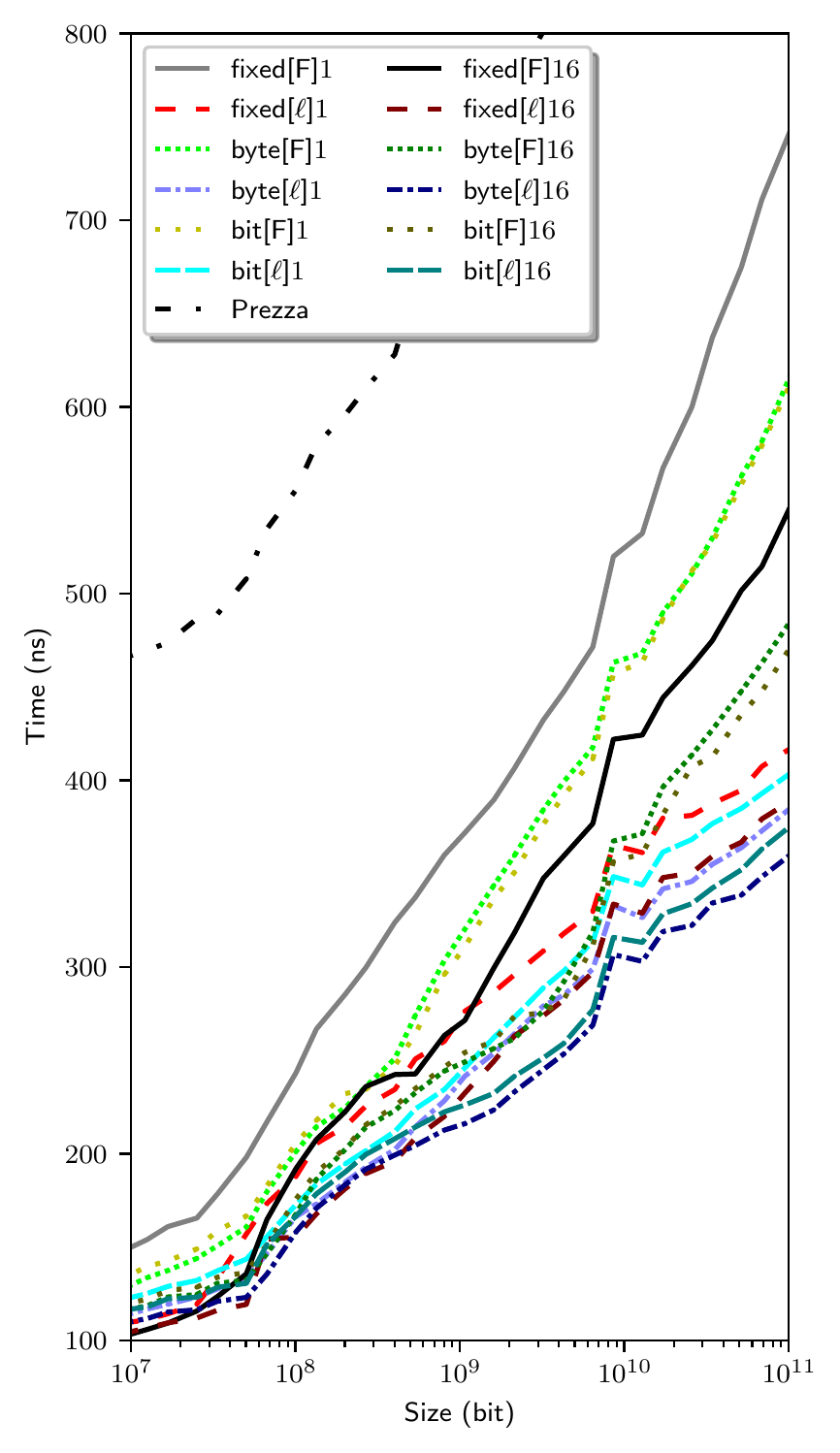}
  \captionof{figure}{\label{fig:bench_select}Performance of $\select$.}
\end{minipage}%
\begin{minipage}{.5\textwidth}
  \centering
  \includegraphics[width=\linewidth]{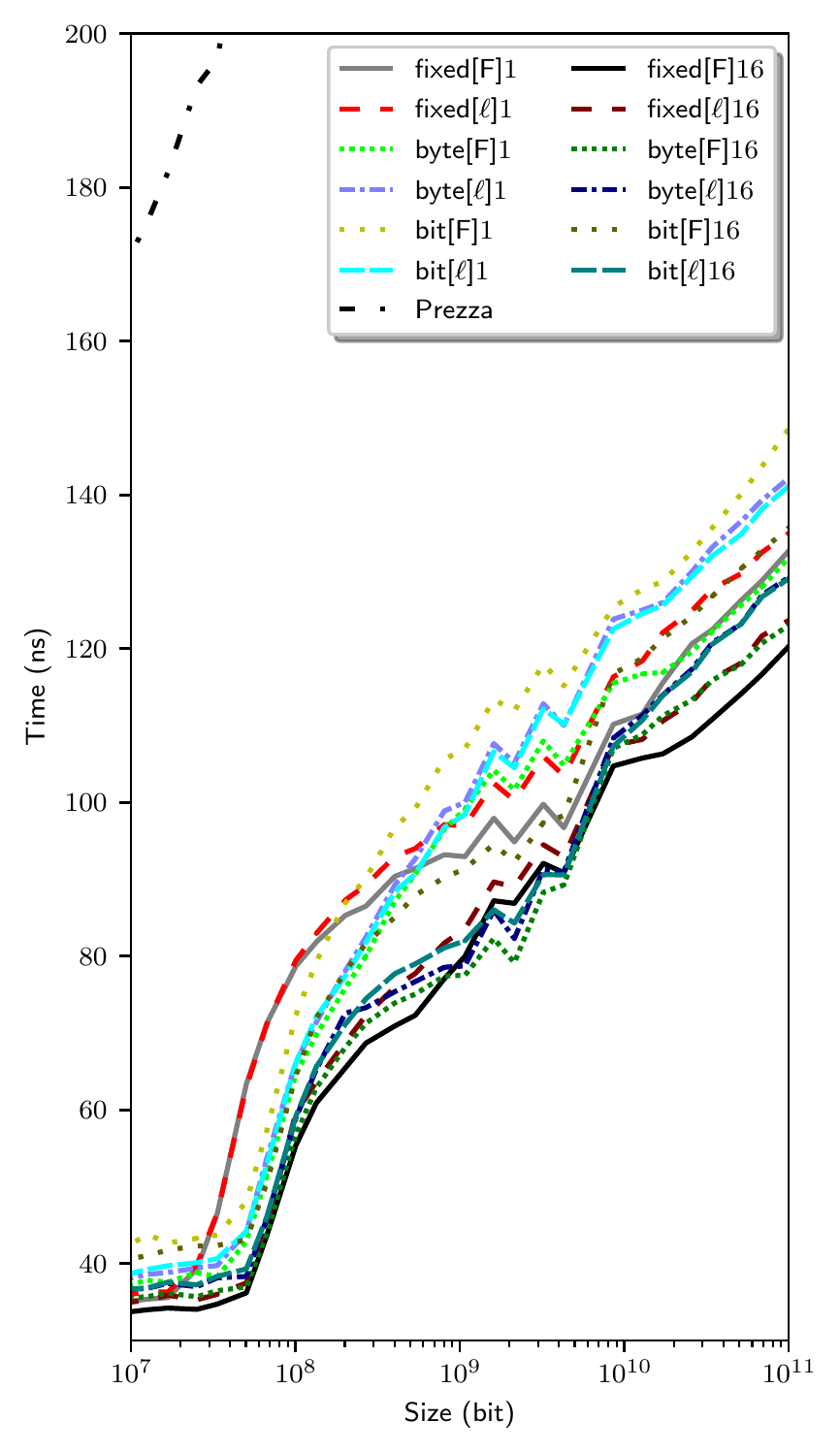}
  \captionof{figure}{\label{fig:bench_update}Performance of bit updates.}
\end{minipage}
\end{figure}

\subsection{Dynamic bit vectors}

The graphs for operations on dynamic bit vectors, show in
Figures~\ref{fig:bench_rank}-\ref{fig:bench_update}, are somewhat smoother than those for
the Fenwick trees, because of the much smaller size (a small fraction of the
size of the bit vector). We report experiment with trivial, one-word blocks, and with blocks make of $16$ words.
Table~\ref{tab:st} reports a space-time tradeoff plot for the same variants at three different sizes.\footnote{Note that the same plot is valid
for the underlying Fenwick tree, modulo some constant offset.}
We did not implement a $\push$/$\pop$ primitive, as its performance would be
essentially constant-time, in case the underlying Fenwick tree does not change,
or identical to $\push$/$\pop$ operations on the tree.
For these experiments we are comparing the performance of our implementations
with Prezza's library for dynamic bit vectors.

\begin{itemize}
  \item When ranking, one-word blocks have a very small advantage at small size (about $5$\,ns), due to the
  simpler code, but the performance penalty on large bit vector is very significant. The performance
  of the underlying Fenwick tree follows closely the results for $\prefix$.
  \item When selecting, lever-order trees with byte and bit compression are the fastest.
  In this case, blocks reduce significantly the number of level of the tree, bringing bit
  compression close to the performance of byte compression.
  \item Updates are the only operation in which fixed-width Fenwick trees win, albeit
  by a very small margin, on the corresponding byte-compressed variant.
\end{itemize}

Table~\ref{tab:space} reports the space usage per bit. With one-word blocks, a fixed-size Fenwick
tree uses the same space of the bit vector, increasing space usage twofold. Bit or byte compression bring
down the space increase to $12$\% and $16$\%, respectively. With 16-word blocks, any Fenwick tree
requires a space that is only a few percent of the bit vector. Prezza's implementation uses $13$\% additional space.

\begin{table}
\centering
\begin{tabular}{rrrrrrrr}
 & \multicolumn{1}{c}{fixed/1} & \multicolumn{1}{c}{byte/1} & \multicolumn{1}{c}{bit/1} & \multicolumn{1}{c}{fixed/16} & \multicolumn{1}{c}{byte/16} & \multicolumn{1}{c}{bit/16} & \multicolumn{1}{c}{Prezza}\\
\hline
\multicolumn{1}{l}{bits} & $2.00$ & $1.16$ & $1.12$ & $1.06$ & $1.02$ & $1.01$ & $1.13$ \\
\hline
\end{tabular}
\caption{\label{tab:space}Average space used per bit on a $10^9$ bit vector. In our case, space depends
 on the block size and on the compression type. Prezza's data structure average space is
  claimed to be around $1.2n$ bits~\cite{PreFDDSSP}, but we are reporting a better
  result obtained by inserting every element at the end (as it would happen in our case).}
\end{table}

\begin{table}
  \centering
  \begin{tabular}{ c | c | c | c | }
    \hline
    \parbox[t]{2mm}{\rotatebox[origin=c]{90}{rank (ns)}}
    & \begin{minipage}{.295\textwidth} \includegraphics[width=\linewidth]{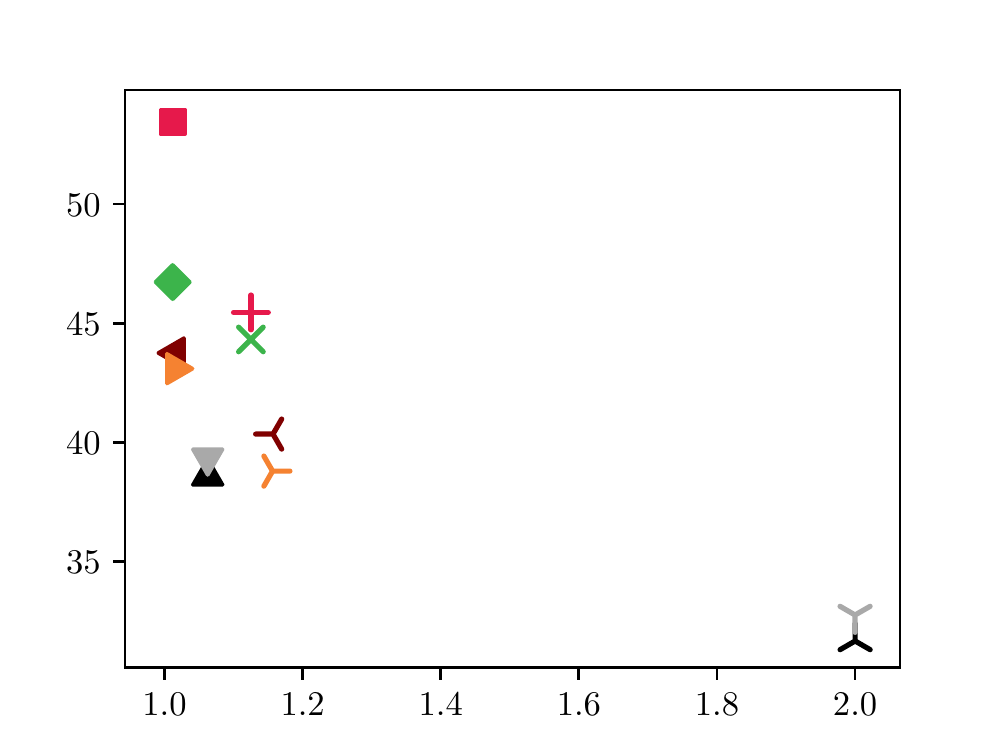} \end{minipage}
    & \begin{minipage}{.295\textwidth} \includegraphics[width=\linewidth]{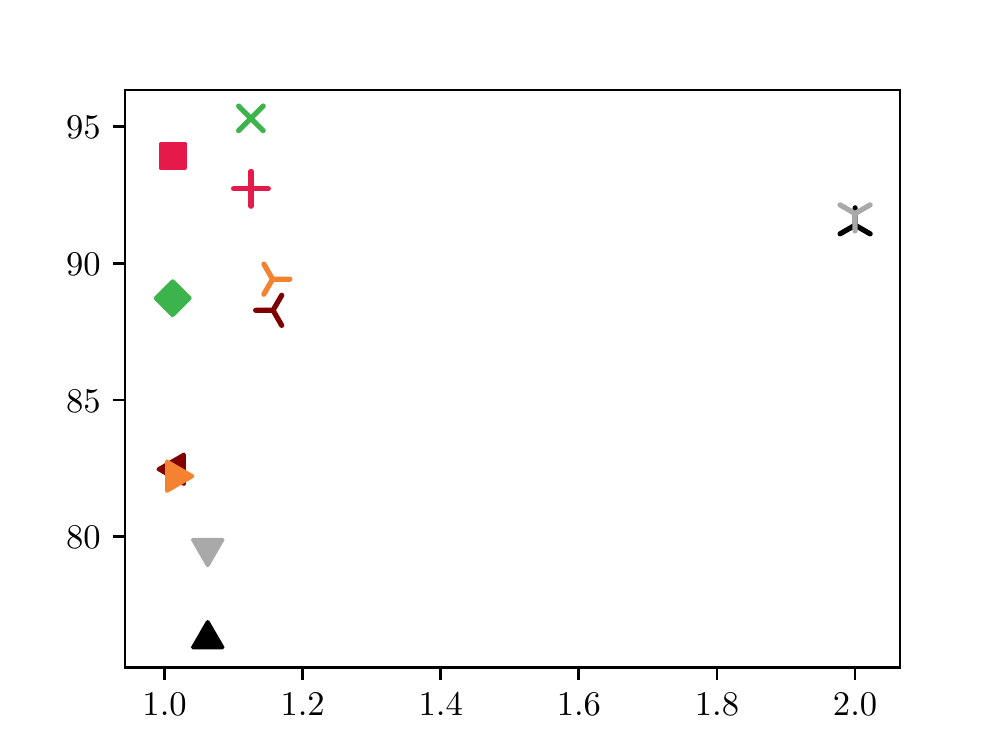} \end{minipage}
    & \begin{minipage}{.295\textwidth} \includegraphics[width=\linewidth]{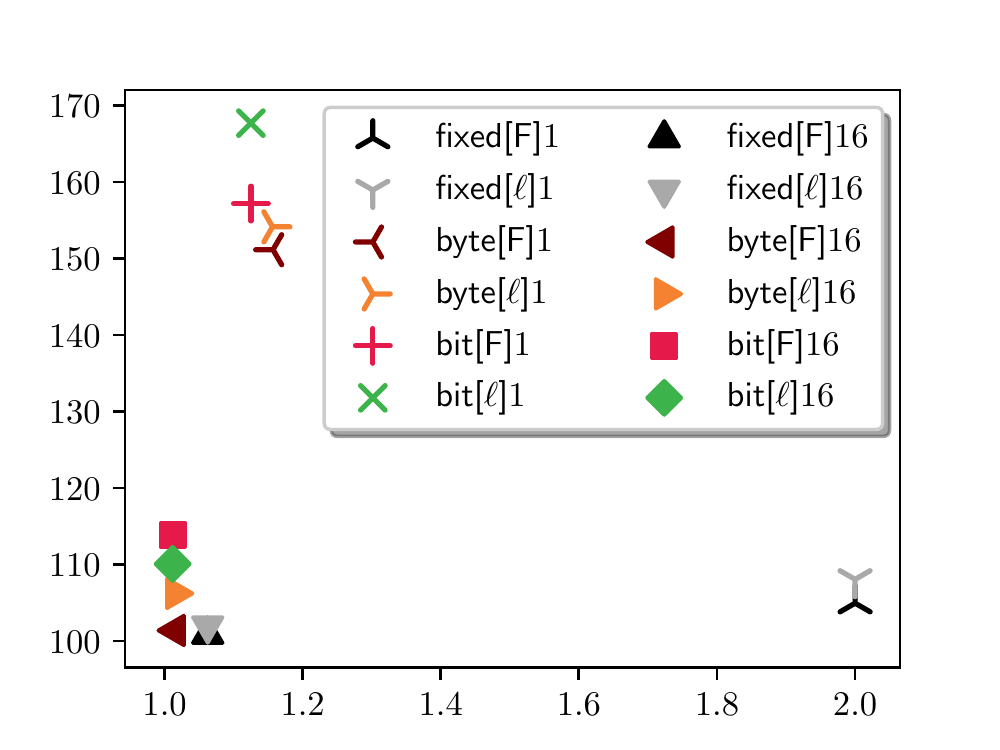} \end{minipage} \\  \hline

    \parbox[t]{2mm}{\rotatebox[origin=c]{90}{select (ns)}}
    & \begin{minipage}{.295\textwidth} \includegraphics[width=\linewidth]{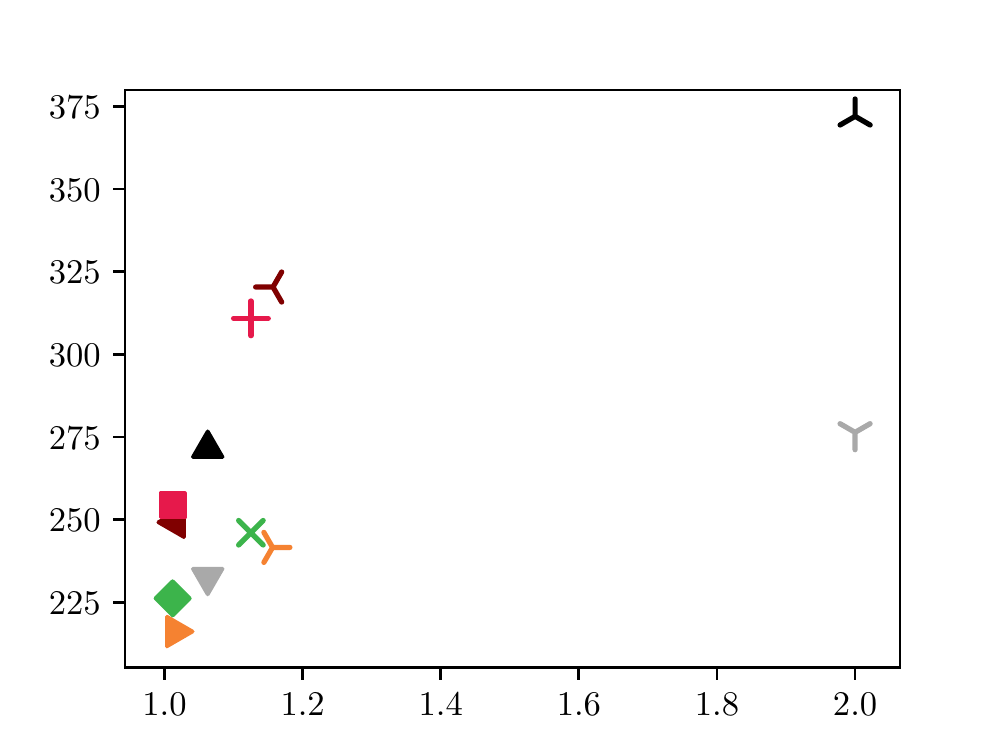} \end{minipage}
    & \begin{minipage}{.295\textwidth} \includegraphics[width=\linewidth]{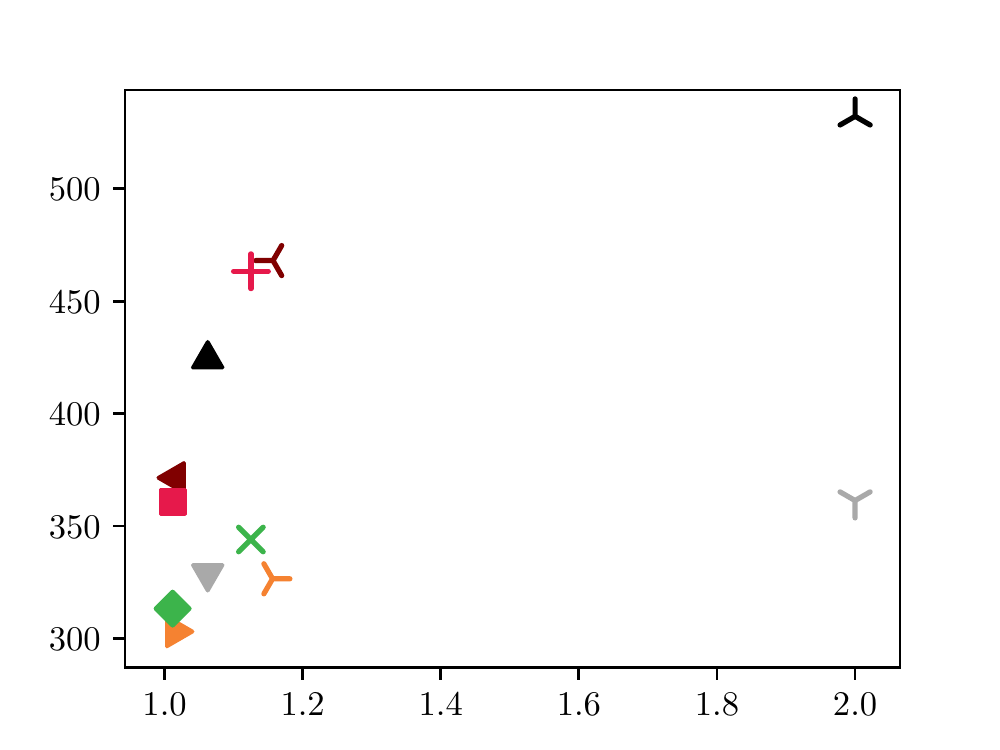} \end{minipage}
    & \begin{minipage}{.295\textwidth} \includegraphics[width=\linewidth]{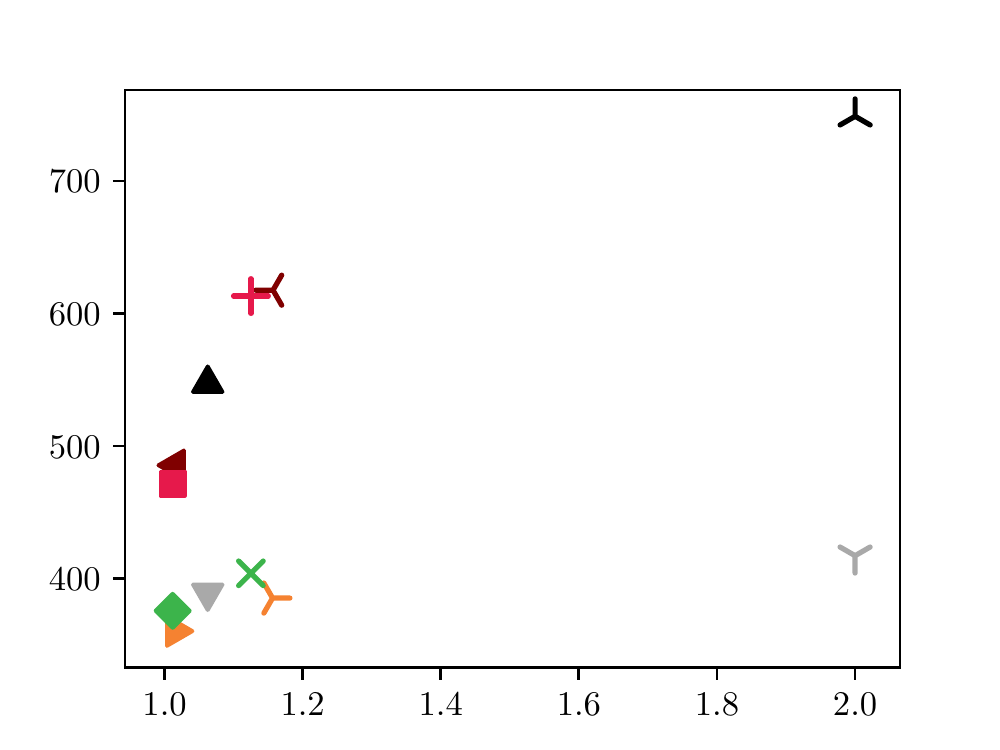} \end{minipage} \\ \hline

    \parbox[t]{2mm}{\rotatebox[origin=c]{90}{update (ns)}}
    & \begin{minipage}{.295\textwidth} \includegraphics[width=\linewidth]{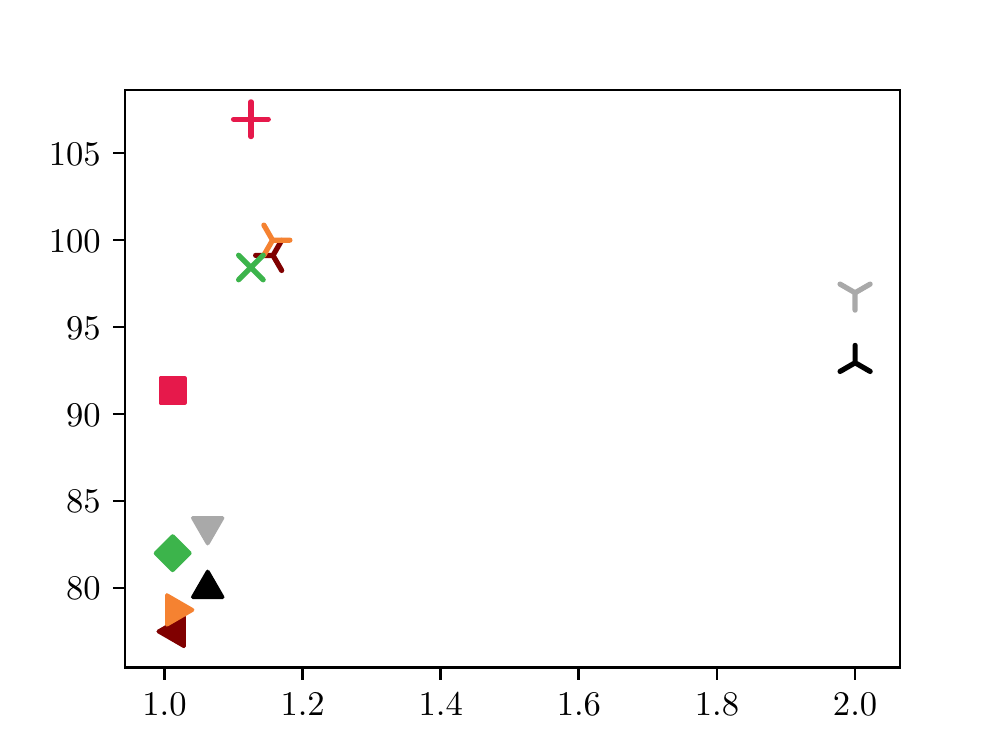} \end{minipage}
    & \begin{minipage}{.295\textwidth} \includegraphics[width=\linewidth]{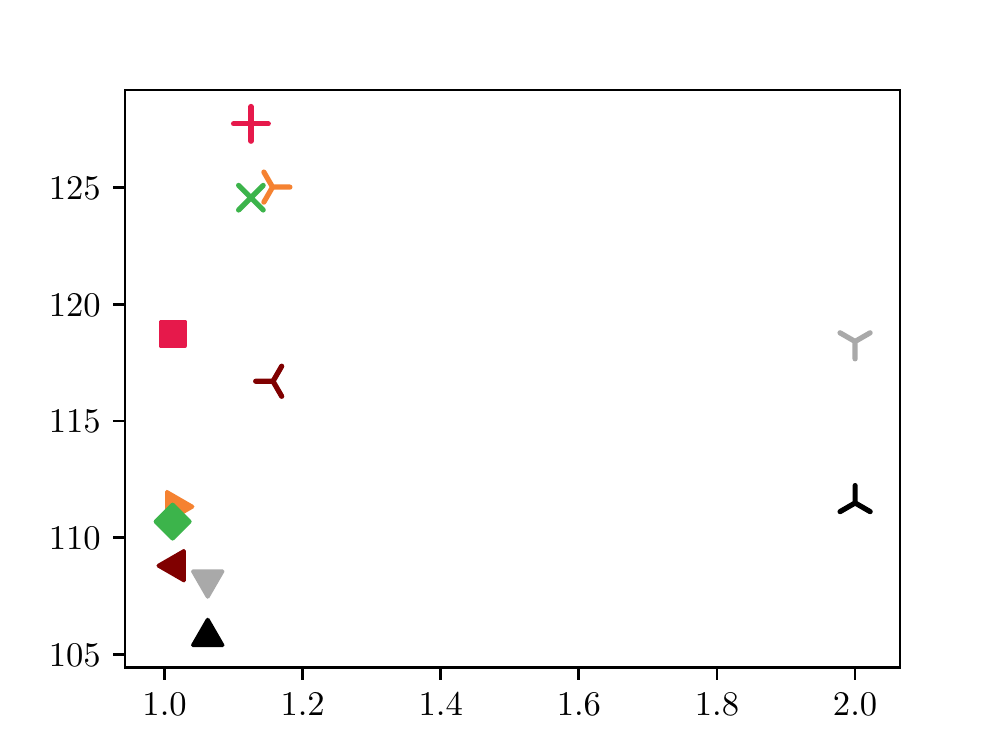} \end{minipage}
    & \begin{minipage}{.295\textwidth} \includegraphics[width=\linewidth]{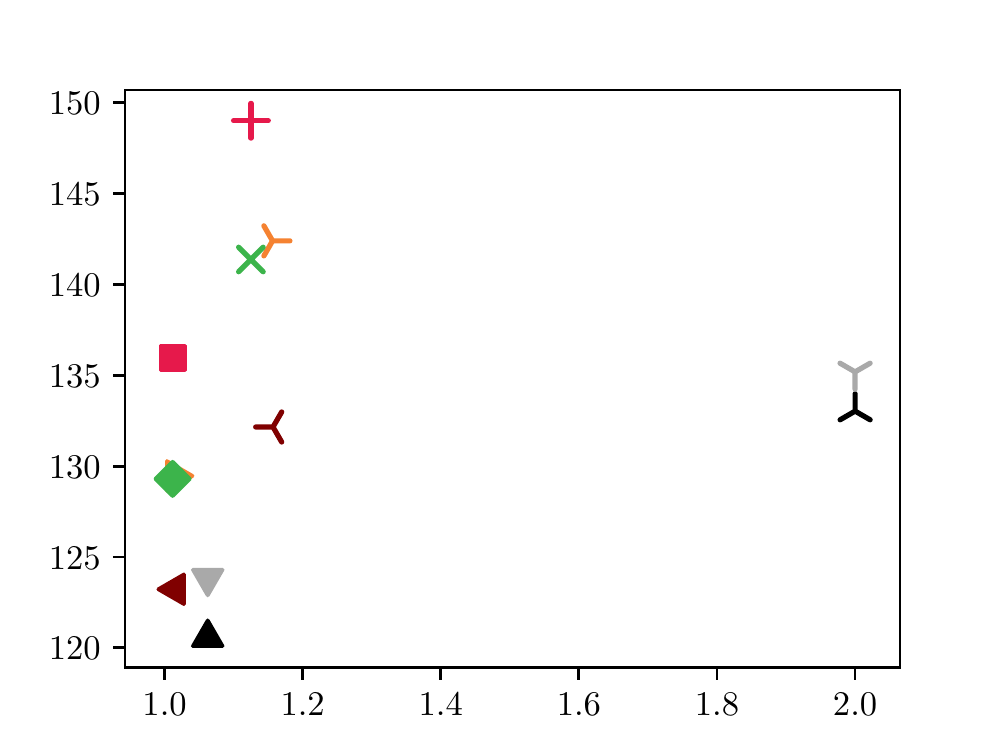} \end{minipage} \\ \hline

    & $10^9$ bits & $10^{10}$ bits & $10^{11}$ bits
  \end{tabular}

  \caption{\label{tab:st}Space-time tradeoff plots for the dynamic bit vector implementations of Figure~\ref{fig:bench_rank}-\ref{fig:bench_update} at different sizes.}
\end{table}

\subsection{Counting transpositions}

In this application, we use our data structures in a classical task: given a permutation
$\pi$, count the number of transposition (e.g., exchanges of two elements) that are necessary to
generate $\pi$. 
This computation is a fundamental step in computing \emph{Kemeny's distance}~\cite{KemMWN} (the
number of transpositions that are necessary to transform a permutation in another), Kendall's $\tau$~\cite{KenNMRC}
(a statistical correlation index), etc. One starts from a permutation $\pi$ on $n$ elements, computes the inverse $\pi^{-1}$, 
initializes a bit vector containing $n$ ones, and then scans $\pi^{-1}$ (seen as a list of numbers): for each element $x$, one computes
the rank at $x$, and then clears the bit of index $x$. In this way, every rank represents the number of elements that $x$
would be exchanged with in a bubble sort of $\pi$; the sum of these ranks is exactly the number of transposition
generating $\pi$.

In Figure~\ref{fig:transp}, we report the time per element to
count transpositions using a standard Fenwick tree, as it happens in SciPy's code,~\cite{JOPSP} 
using Prezza's library and using a fixed-size and a byte-compressed implementation of our
dynamic ranking structure. To maximize the difference between each structure, we generated
$\pi^{-1}$ implicitly using a suitable linear congruential pseudorandom number generator with power-of-two modulus\footnote{We applied 
a bijection to the output of the generator to prevent the short period of the low bits to influence the results.} Besides
the evident speed increase, our structures (and Prezza's) are more than fifty times smaller than SciPy's.
\begin{figure}
  \centering
  \includegraphics[width=0.5\linewidth]{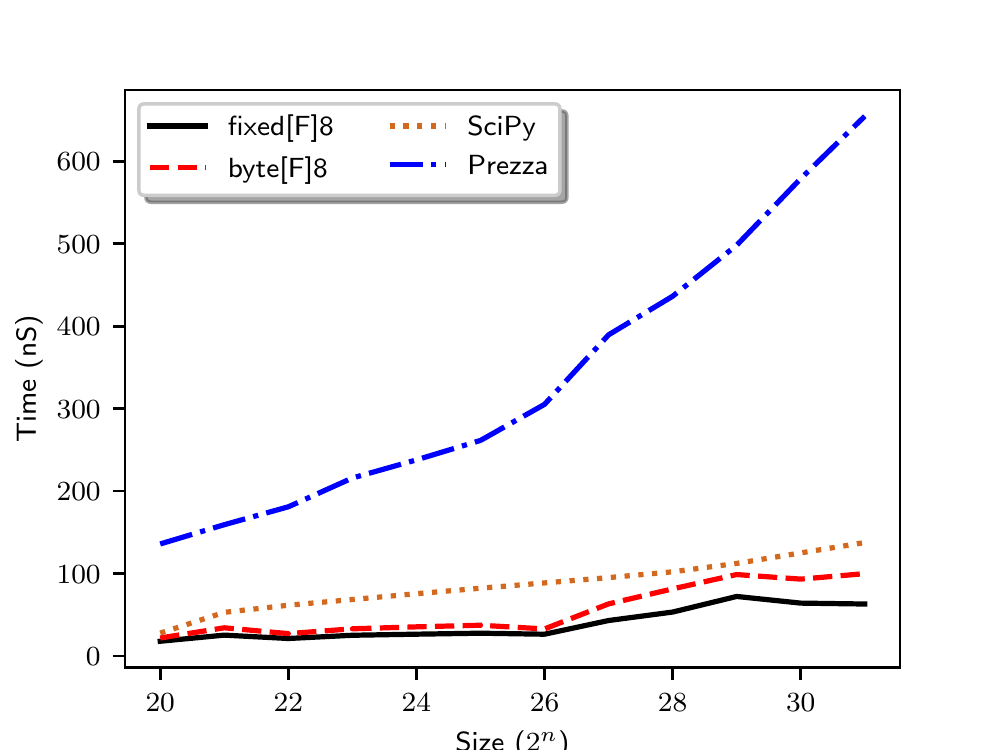}
  \caption{\label{fig:transp}Counting transpositions of a random permutation. Note that the SciPy code uses
  more than fifty times more memory than the alternative structures.}
\end{figure}

\subsection{Generating graphs by preferential attachment}

A simple, classical model of random (undirected) graphs is the \emph{preferential attachment} model:
one fixes an integer $d$, a $d_0\geq d$, and starts
with $d_0$ isolated vertices (with a self-loop). Then, at each iteration a vertex of degree $d$ is added
to the graph, and connected to $d$ vertices chosen among the ones already generated in a way that is
proportional to their degree.
If enough memory is available, a simple way to sample vertices proportionally to the degree is
to maintain a list of vertices in which each vertex is repeated as many times as its degree, and
then just sample from the list. For large graphs, however, the list becomes unmanageable.

The Fenwick tree offers a simple solution: one maintains a tree representing a list $\bm v$ containing the degree 
of each vertex, and at that point a \textsf{find} operation on a uniform random sample in the
integer interval $[0\..2m)$, where $m$ is the current number of edges of the graph, samples vertices proportionally to their degree. 
Thus, we alternate \textsf{find} operations, to find vertices with which to connect, \textsf{add} operations, 
to update their degree, and $\push$ operations, one for each new vertex.  

In Figure~\ref{fig:ab}, we report the time necessary to generate a graph of a million vertices using 
four variants of our data structures for increasing $d$. The level-order layout is a clear winner, as
predicted by the better cache behavior of $\find$. Note that the byte-compressed version is slower
of the fixed-size version, but it uses much less memory, as $B=dn$. This fact is apparently in contradiction
with Figure~\ref{fig:bench_find}: however, in the sampling process the nodes traversed by
\textsf{find} operations are strongly biased towards nodes that lead to vertices of high degree,
so the advantage of the simpler code of the fixed-size version outweighs its larger memory footprint.

\begin{figure}
  \centering
  \includegraphics[width=0.5\linewidth]{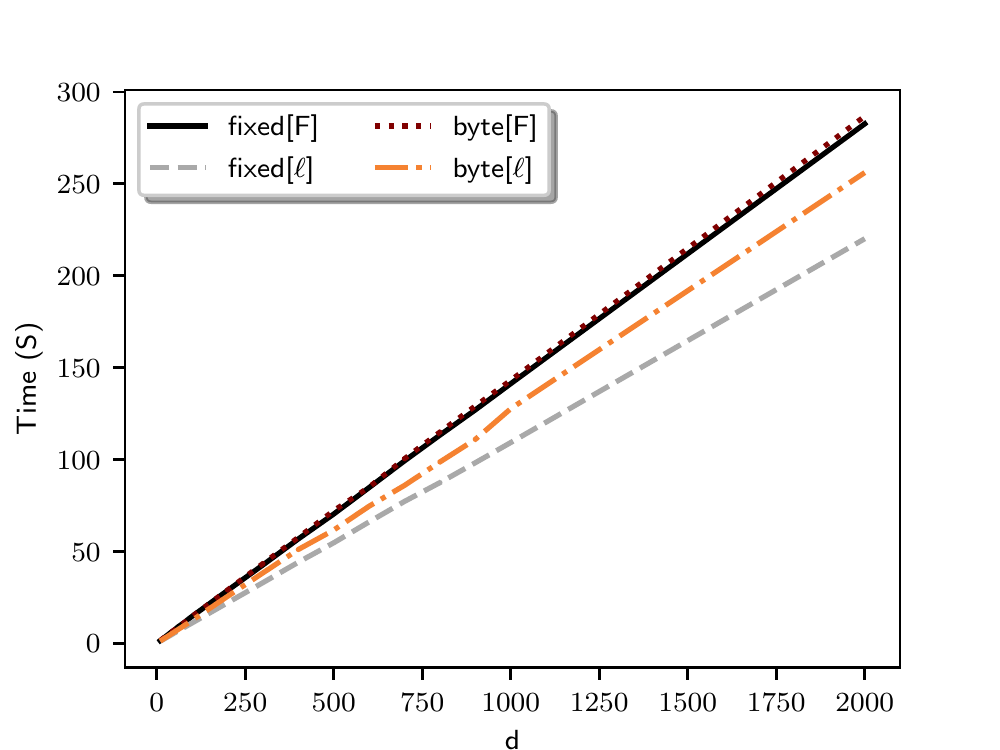}
  \caption{\label{fig:ab}Generating preferential-attachment pseudorandom graphs with a million vertices and average degree $d$ using Fenwick trees. Level-order
  layout has a clear advantage.}
\end{figure}

\section{Conclusions}

We have presented improved, cache-friendly and prediction-friendly variants of
the classical Fenwick tree. Besides maintaining a prefix-sum data structure, the
tree can be used to provide an efficient dynamic bit vector with selection and
ranking with a very small space overhead, albeit size can be changed only
by adding or removing a bit at the end, rather than at an arbitrary position.

The first takeaway lesson from our study is that it is fundamental to perturb
the structure of a Fenwick tree in classical Fenwick layout so that it does not
interfere with the inner working of multi-way caches: while this phenomenon has
never been reported before, the interference can lead to a severe
underperformance of the data structure.

The second lesson is that, in spite of its elegance, the Fenwick
layout is advantageous only if the $\find$ primitive is never used (e.g., when
counting transpositions). In all other cases, a level-order layout is preferable.

Bit-compressed versions are extremely tight, but the higher access time makes them
palatable only when memory is really scarce: otherwise, byte-compressed versions
provide often similar occupancy, but a much faster access. In particular, the faster
$\find$ results we report are for a byte-compressed, level-ordered tree.

For what matter dynamic bit vectors, for the sizes we consider the best solution is a
block large as one or two cache lines; the consideration made
for Fenwick order vs.~level order and compression are valid also in this case.

Albeit not reported in the paper, we also experimented with the idea of
\emph{hybrid} trees, which upper levels are level ordered, but lower
lever have a Fenwick layout. In some architectures, and in particular if only
small memory pages are available, they offer a performance that is intermediate between
that of Fenwick layout and level-order layout. In general, for the upper
levels one can always use a fixed-width Fenwick tree, as the number of nodes involved
is very small, whereas lower levels benefit from compression.

Finally, we remark once again that due to the highly nonlinear effect of cache
architecture, and to conflicting cache usage, benchmarking a target application
with different variants is the best way to choose the variant that better suits the application.

\bibliography{biblio}

\end{document}